\documentclass[aps,pra,final,twocolumn,letterpaper,superscriptaddress,preprintnumbers,accepted=2020-04-17]{Quantum_Template/quantumarticle}
\pdfoutput=1

\usepackage[utf8]{inputenc}
\usepackage[english]{babel}
\usepackage[T1]{fontenc}
\usepackage{amsmath}
\usepackage{hyperref}

\usepackage{graphicx}                            % Graphics package
\usepackage{epstopdf}

\usepackage{amssymb}
\usepackage{float}
\usepackage{amsthm}
\usepackage{enumerate}
\usepackage{qcircuit}

\usepackage{complexity}
\usepackage[usenames,dvipsnames]{color}
\usepackage[normalem]{ulem}
\usepackage{tikz}
\usetikzlibrary{arrows}
\usepackage[all]{nowidow}
\usepackage[numbers,sort&compress]{natbib}
\usepackage{dsfont}

\usepackage{xcolor}
\newcommand\hl[1]{#1}

\newcommand{\bra}[1]{\left\langle #1 \right|}
\newcommand{\ket}[1]{\left|#1\right\rangle}
\newcommand{\braket}[2]{\left\langle#1 |  #2\right\rangle}

\newcommand{\be}{\begin{equation}}
\newcommand{\ee}{\end{equation}}

\newtheorem{theorem}{Theorem}
\newtheorem{corollary}{Corollary}[theorem]
\newtheorem{lemma}[theorem]{Lemma}
\newtheorem{conj}{Conjecture}

\newtheorem{innercustomthm}{Theorem}
\newenvironment{customthm}[1]
  {\renewcommand\theinnercustomthm{#1}\innercustomthm}
  {\endinnercustomthm}
  
\def\nn{\nonumber\\}

\DeclareMathOperator{\gap}{gap}
\DeclareMathOperator{\ngap}{ngap}
\DeclareMathOperator{\Per}{Per}
\DeclareMathOperator\erf{erf}
\newcommand{\nc}{\newcommand}
\nc{\rnc}{\renewcommand}

\nc\cF{\mathcal{F}}

\nc\bbC{\mathbb{C}}

\nc\bbF{\mathbb{F}}
\nc\bbM{\mathbb{M}}
\nc\bbN{\mathbb{N}}
\nc\bbR{\mathbb{R}}
\nc\bbZ{\mathbb{Z}}

\nc\ptNc{a}
\nc\ptSBc{a'}
\nc\piNc{b}
\nc\piSBc{b'}

\begin{document}

\preprint{MIT-CTP/5019}

\title{How many qubits are needed for quantum computational supremacy?}
\author{Alexander M. Dalzell}\email{adalzell@caltech.edu}
\affiliation{Institute for Quantum Information and Matter, California Institute of Technology, Pasadena, CA 91125, USA}
\affiliation{Department of Physics, Massachusetts Institute of Technology, Cambrdige, Massachusetts 02139, USA}
\orcid{0000-0002-3756-8500}
\author{Aram W. Harrow}\email{aram@mit.edu}\affiliation{Center for Theoretical Physics, Massachusetts Institute of Technology, Cambridge, Massachusetts 02139, USA}
\orcid{0000-0003-3220-7682}
\author{Dax Enshan Koh}\email{dax.koh@zapatacomputing.com}\affiliation{Department of Mathematics, Massachusetts Institute of Technology, Cambridge, Massachusetts 02139, USA}
\affiliation{Zapata Computing, Inc., 100 Federal Street, 20th Floor, Boston, Massachusetts 02110, USA}
\orcid{0000-0002-8968-591X}
\author{Rolando L. La Placa}\email{rlaplaca@mit.edu}\affiliation{Center for Theoretical Physics, Massachusetts Institute of Technology, Cambridge, Massachusetts 02139, USA}
\orcid{0000-0001-7867-226X}

%\date{\today} 
\begin{abstract}

Quantum computational supremacy arguments, which describe a way for a quantum computer to perform a task that cannot also be done by a classical computer, typically require some sort of computational assumption related to the limitations of classical computation. One common assumption is that the polynomial hierarchy ($\mathsf{PH}$) does not collapse, a stronger version of the statement that $\mathsf{P} \neq \mathsf{NP}$, which leads to the conclusion that any classical simulation of certain families of quantum circuits requires time scaling worse than any polynomial in the size of the circuits. However, the asymptotic nature of this conclusion prevents us from calculating exactly how many qubits these quantum circuits must have for their classical simulation to be intractable on modern classical supercomputers. We refine these quantum computational supremacy arguments and perform such a calculation by imposing fine-grained versions of the non-collapse conjecture. \hl{Our first two conjectures poly3-NSETH($\ptNc$) and per-int-NSETH($\piNc$) take specific classical counting problems related to the number of zeros of a degree-3 polynomial in $n$ variables over $\mathbb{F}_2$ or the permanent of an $n \times n$ integer-valued matrix, and assert that any non-deterministic algorithm that solves them requires $2^{cn}$ time steps, where $c \in \{\ptNc,\piNc\}$. A third conjecture poly3-ave-SBSETH($\ptSBc$) asserts a similar statement about average-case algorithms living in the exponential-time version of the complexity class $\SBP$. We analyze evidence for these conjectures and argue that they are plausible when $\ptNc=1/2$, $\piNc = 0.999$ and $\ptSBc = 1/2$. 

%\vspace{12 pt}

Imposing poly3-NSETH(1/2) and per-int-NSETH(0.999), and assuming that the runtime of a hypothetical quantum circuit simulation algorithm would scale linearly with the number of gates/constraints/optical elements, we conclude that Instantaneous Quantum Polynomial-Time (IQP) circuits with 208 qubits and 500 gates, Quantum Approximate Optimization Algorithm (QAOA) circuits with 420 qubits and 500 constraints and boson sampling circuits (i.e.~linear optical networks) with 98 photons and 500 optical elements are large enough for the task of producing samples from their output distributions up to constant multiplicative error to be intractable on current technology. Imposing poly3-ave-SBSETH(1/2), we additionally rule out simulations with constant additive error for IQP and QAOA circuits of the same size. Without the assumption of linearly increasing simulation time, we can make analogous statements for circuits with slightly fewer qubits but requiring $10^4$ to $10^7$ gates.}

\end{abstract}

\maketitle

%\pagestyle{myheadings}
%\markboth{}{}
%\thispagestyle{empty}

\section{Introduction}\label{sec:introduction}

Quantum computational supremacy (QCS) is the goal of carrying out a computational task on a quantum computer that cannot be performed by any classical computer \cite{preskill2012quantum}. Ingredients of this include choosing an appropriate task, building a quantum device that can perform it, ideally verifying that it was done correctly, and finally using arguments from complexity theory to support the claim that no classical computer can do the
same \cite{harrow2017quantum}. Recent advances indicate that the experimental ingredient might be available \hl{very soon --- indeed, Google recently reported \cite{arute2019quantum} that it has attained QCS with a 53 qubit superconducting device (later we comment more on how this announcement fits in with our analysis) --- } but the choice of task, its verification, and its complexity-theoretic justification remain important open theoretical research questions. In particular, based on the current status of complexity theory, establishing limitations on classical computing for the purpose of assessing how close we are to demonstrating QCS requires making conjectures, and thus we are presented with a range of choices. If we make stronger conjectures then we can use a smaller and more restricted quantum computer while ruling out the existence of more powerful classical simulation algorithms. Weaker conjectures, on the other hand, are more defensible and can be based on more widely studied mathematical principles.

A leading example of a strong conjecture is the Quantum Threshold Assumption (QUATH) proposed by Aaronson and Chen \cite{aaronson2016complexity}, which states that 
there is no efficient (i.e.~polynomial-time) classical algorithm that takes as input a description of a random quantum circuit $C$, and decides whether $\lvert\bra{ 0^n } C \ket{0^n}\rvert^2$ is greater or less than the median of all $\lvert\bra{ 0^n } C \ket{0^n}\rvert^2$ values, with success probability at least $\tfrac 12 + \Omega(\tfrac 1{2^n})$ over the choice of $C$.
This conjecture gives one of the strongest possible statements about the hardness of simulating quantum circuits that is not already ruled out by known simulations. 

A weaker conjecture is the statement that the polynomial hierarchy ($\PH$) does not collapse, which is closely related to the assertion that $\P \neq \NP$. Under this assumption, it has been shown that there cannot exist an efficient classical algorithm to produce samples from the output distribution of certain families of quantum circuits \cite{terhal2002adaptive,aaronson2011computational,bremner2010classical,jozsa2014classical, koh2015further, morimae2014hardness, fujii2014impossibility, farhi2016quantum, bouland2016complexity, aaronson2016computational,bouland2017quantum,bermejo2017architectures,hangleiter2017anti,morimae2018merlin}, up to constant multiplicative error. The three families we focus on in this work are Instantaneous Quantum Polynomial-time (IQP) circuits \cite{shepherd2009temporally,bremner2010classical}, Quantum Approximate Optimization Algorithm (QAOA) circuits \cite{farhi2014quantum,farhi2016quantum}, and boson sampling circuits (i.e.~linear optical networks) \cite{aaronson2011computational}, all of which are among those whose simulation is hard for the \PH. Indeed, a key selling point for work in QCS is that it could be based not on the conjectured hardness of a particular quantum circuit family or even quantum mechanics in general, but instead on highly plausible, purely classical computational conjectures, such as the non-collapse of the \PH.

However, the non-collapse of the $\PH$ is in a sense too weak of a conjecture to be practically useful. The conjecture rules out polynomial-time simulation algorithms for these families of circuits, but does not describe a concrete superpolynomial lower bound. Thus, assuming only the non-collapse of the $\PH$ would be consistent with a simulation of an $n$-qubit quantum system running in time $n^{f(n)}$ for an arbitrarily slowly growing function $f(n)$, say $\log\log\log\log(n)$. A stronger conjecture might lead to a requirement that simulation algorithms be exponential time, meaning that there is some constant $c$ for which its runtime is $\geq 2^{cn}$. Even this, though, is not strong enough; it remains possible that the constant $c$ is sufficiently small that we cannot rule out a scenario where highly parallelized state-of-the-art classical supercomputers, which operate at as many as $10^{17}-10^{18}$ floating-point operations per second (FLOPS), are able to simulate any circuit that might be experimentally realized in the near-term. For example, Neville et al.~\cite{neville2017classical}, as well as Clifford and Clifford \cite{clifford2018classical} recently developed classical algorithms that produce samples from the output of boson sampling circuits, the former of which has been shown to simulate $n=30$ photons on a standard laptop in just half an hour, contradicting the belief of many that $20$ to $30$ photons are sufficient to demonstrate a definitive quantum advantage over classical computation. A stronger conjecture that restricts the value of the exponential factor $c$, a so-called ``fine-grained'' conjecture, is needed to move forward on assessing the viability of QCS protocols. The framework of fine-grained complexity has gathered much interest in its own right in the last decade (see \cite{williams2015hardness} for survey), yielding unexpected connections between the fine-grained runtime of solutions to different problems.

In this work, we examine existing QCS arguments for IQP, QAOA, and boson sampling circuits from a fine-grained perspective. While many previous arguments \cite{aaronson2011computational,bremner2010classical,farhi2016quantum} center on the counting complexity class $\PP$, which can be related to quantum circuits via postselection \cite{aaronson2005quantum}, the fine-graining process runs more smoothly when we instead use the counting class $\co\C_=\P$, \hl{which was first utilized in the context of QCS in \cite{fujii2014impossibility,fujii2015power}}. The class $\co\C_=\P$ is the set of languages for which there exists an efficient classical probabilistic algorithm that accepts with probability exactly 1/2 only on inputs not in the language. It can be related to quantum circuits via non-determinism: $\co\C_=\P = \NQP$ \cite{fenner1999determining}, where $\NQP$, a quantum analogue of $\NP$, is the class of languages for which there exists an efficient quantum circuit that has non-zero acceptance probability only on inputs in the language. Moreover, this equality still holds when we restrict $\NQP$ to quantum computations with IQP, QAOA, or boson sampling circuits. Additionally, it is known that if $\co\C_=\P$ were to be equal to $\NP$, the $\PH$ would collapse to the second level \cite{toda1992counting,fenner1999determining}. Thus, by making the assumption that there is a problem in $\co\C_=\P$ that does not admit a non-deterministic polynomial-time solution, i.e.~$\co\C_=\P \not\subset \NP$, we conclude that there does not exist a classical simulation algorithm that samples from the output distribution of IQP or QAOA circuits up to constant multiplicative error, for this would imply $\NP = \NQP = \co\C_=\P$, contradicting the assumption.

To make a fine-grained version of this statement, we pick a specific $\co\C_=\P$-complete problem related to the number of zeros of degree-3 polynomials over the field $\mathbb{F}_2$, which we call $\texttt{poly3-NONBALANCED}$, and we assume that
\texttt{poly3-NONBALANCED} does not have a non-deterministic algorithm running in fewer than $T(n)$ time steps for an explicit function $T(n)$. We choose $T(n)=2^{\ptNc n-1}$ for a fixed constant $\ptNc $ and call this conjecture the degree-3 polynomial Non-deterministic Strong Exponential Time Hypothesis (poly3-NSETH($\ptNc $)). It is clear that poly3-NSETH($\ptNc $) is false when $\ptNc  > 1$ due to the brute-force deterministic counting algorithm that iterates through each of the $2^n$ possible inputs to the function $f$. However, a non-trivial algorithm by Lokshtanov,  Paturi, Tamaki, Williams and Yu (LPTWY)~\cite{lokshtanov2017beating} gives a better-than-brute-force, deterministic algorithm for counting zeros to systems of degree-$k$ polynomial that rules out poly3-NSETH($\ptNc $) whenever $\ptNc  > 0.9965$.  It may be possible to improve this constant while keeping the same basic method but, as we discuss in Appendix~\ref{app:algo}, we expect any such improvements to be small.  Refuting poly3-NSETH($\ptNc $) for values of $\ptNc $ substantially below 1 would require the development of novel techniques. 

Assuming poly3-NSETH($\ptNc $), in Section \ref{sec:multerror} we derive a fine-grained lower bound on the runtime for any \hl{multiplicative-error} classical simulation algorithm for QAOA and IQP circuits with $n$ qubits. In essence, what we show is that a classical simulation algorithm that beats our lower bounds could be used as a subroutine to break poly3-NSETH($\ptNc $).  Then, we repeat the process for boson sampling circuits with $n$ photons by replacing poly3-NSETH($\ptNc $) with a similar conjecture we call per-int-NSETH($\piNc $) involving the permanent of $n \times n$ integer-valued matrices. In this case, however, there is no known algorithm that can rule out any values of $\piNc $ when $\piNc < 1$. Accordingly, the lower bound we derive on the simulation time of boson sampling circuits when we take $\piNc  = 0.999$ is essentially tight, matching the runtime of the naive simulation algorithm up to factors logarithmic in the total runtime. 
Recently, a similar approach was applied to obtain lower bounds on the difficulty of computing output probabilities of quantum circuits based on the SETH conjecture~\cite{monotone-sim-LB}.  Our work has the disadvantage of using a less well-studied and possibly stronger conjecture (poly3-NSETH($\ptNc $)) but the advantage of ruling out classical algorithms for sampling, i.e.~for the same tasks performed by the quantum computer. In Section \ref{sec:evidence}, we discuss evidence for our conjectures poly3-NSETH($\ptNc$) and per-int-NSETH($\piNc$), and discuss their relationship to other proposed fine-grained conjectures.

\hl{However, realistic near-term quantum devices, which are subject to non-negligible noise, sample from a distribution that differs from ideal with constant \textit{additive} error, not constant multiplicative error, and ruling out additive-error classical simulation algorithms presents additional technical challenges. Previous work ruled out polynomial-time simulation algorithms of this type by conjecturing the non-collapse of the $\PH$ and additionally that certain problems we know to be hard for the $\PH$ in the worst case are also hard in the average case \cite{aaronson2011computational,bremner2016average,bouland2018quantum,bouland2017quantum}. In Section \ref{sec:adderror}, we present an argument that also rules out some exponential-time additive-error simulation algorithms for IQP and QAOA circuits based on a fine-grained conjecture we call poly3-ave-SBSETH($\ptSBc$). Our analysis may be viewed generally as a fine-grained version of previous work; however, our approach has crucial novel elements that go beyond simply being fine-grained. Unlike previous work on additive-error QCS, our analysis bypasses the need for Stockmeyer's theorem (which would incur significant fine-grained overhead) and as a result it is perhaps conceptually simpler. Additionally, in order to give evidence for poly3-ave-SBSETH($\ptSBc$), which is both fine-grained and average-case, we make several technical observations which could be of independent interest. We prove an average-case black-box query lower bound, and we give a self-reduction from the worst case to a hybrid average/worst-case for the problem of computing the number of zeros of a degree-3 polynomial over $\mathbb{F}_2$. Along the way, we also provide a more complete picture of the distribution of the number of zeros for uniformly random degree-3 polynomials over $\mathbb{F}_2$ --- extending work that showed the distribution anticoncentrates \cite{bremner2016average} --- by formally proving that all its moments approach those of a Gaussian distribution as the number of variables in the polynomial increases.}

Finally in Section \ref{sec:numqubits}, we show how these lower bounds lead to estimates for the number of qubits that would be sufficient for quantum supremacy under our conjectures. We conclude that classically simulating (up to multiplicative error) general IQP circuits with $93/\ptNc $ qubits, QAOA circuits with $185/\ptNc $ qubits, or boson sampling circuits with $93/\piNc $ photons would require one century for today's fastest supercomputers, which we consider to be a good measure of intractability. \hl{For additive error we may replace $\ptNc$ with $\ptSBc$.} We believe values for $\ptNc$, $\ptSBc$, and $\piNc $, leading to plausible conjectures are $\ptNc = 1/2$, which is substantially below best known better-than-brute-force algorithms, \hl{$\ptSBc=1/2$, which matches best known algorithms}, and $\piNc  = 0.999$, which is roughly equivalent to asserting that the best known brute force algorithm is optimal up to subexponential factors. The relative factor of two in the number of qubits for QAOA circuits comes from a need for ancilla qubits in constructing a QAOA circuit to solve the \texttt{poly3-NONBALANCED} problem. However, these circuits must have $10^4$ to $10^7$ gates for these bounds to apply. \hl{Under the additional assumption that the complexity of any simulation algorithm would scale linearly with the number of gates, we conclude that circuits with only 500 gates and 208 IQP qubits, 420 QAOA qubits, or 98 photons would be sufficient for QCS.}  By comparison, factoring a 1024-bit integer, which is sufficiently beyond the capabilities of today's classical computers running best known algorithms, has been estimated to require more than 2000 qubits and on the order of $10^{11}$ gates using Shor's algorithm \cite{roetteler2017quantum}.

\section{Background}
\subsection{Counting complexity and quantum computational supremacy}

The computational assumptions underlying our work and many previous QCS results utilize a relationship between quantum circuits and counting complexity classes that is not seen to exist for classical computation. To understand this relationship, we quickly review several definitions and key results.

Let $n \geq 1$, and $f:\{0,1\}^n \rightarrow \{0,1\}$ be a Boolean function. The \textit{gap} of $f$ is defined to be
\begin{equation}\label{eq:gapdef}
\gap(f)= \sum_{x\in \{0,1\}^n} (-1)^{f(x)}.
\end{equation}

Note that the number of zeros of $f$ may be written in terms of the gap, as follows:
\begin{equation}
|\{x\in \{0,1\}^n:f(x)=0\}|= \tfrac 12(2^n+ \gap (f)).
\end{equation}

Various complexity classes may be defined in terms of the gap.
The class $\#\P$ is defined to be the class of functions $f:\{0,1\}^* \rightarrow \mathbb N$ for which there exists a polynomial $p$ and a polynomial-time Turing machine $M$ such that for all $x \in \{0,1\}^*$,
\begin{eqnarray}
f(x) &=& |\{ y \in \{0,1\}^{p(|x|)}: M(x,y) = 0\}| \nn
&=&
\frac 12(2^{p(|x|)} + \gap( M(x,\cdot))).
\end{eqnarray}
Thus, $\#\P$ contains functions that \textit{count} the number of zeros of a polynomial-time computable Boolean function.

A language $L$ is in $\PP$ if there exists a polynomial $p$ and a polynomial-time Turing machine $M$ such that for all $x \in \{0,1\}^*$,
\begin{eqnarray}
x \in L  &\iff & |\{ y \in \{0,1\}^{p(|x|)}: M(x,y) = 0\}| \nn
&&< |\{ y \in \{0,1\}^{p(|x|)}: M(x,y) = 1\}| \nn
&\iff & \gap( M(x,\cdot))<0.
\end{eqnarray}
  
The class $\NP$ is defined similarly, but where
\begin{eqnarray}
x \in L  &\iff & |\{ y \in \{0,1\}^{p(|x|)}: M(x,y) = 1\}| \neq 0 \nn
&\iff & \gap( M(x,\cdot)) \neq 2^{p(|x|)},
\end{eqnarray}
and the class $\co\C_=\P$, where
\begin{eqnarray}
x \in L  &\iff & |\{ y \in \{0,1\}^{p(|x|)}: M(x,y) = 0\}| \nn &&\neq |\{ y \in \{0,1\}^{p(|x|)}: M(x,y) = 1\}| \nn
&\iff & \gap( M(x,\cdot)) \neq 0.
\end{eqnarray}

By interpreting $M$ as a probabilistic algorithm and $y$ as the random string of bits used by $M$, we can redefine $\NP$, $\PP$, and $\co\C_=\P$ as the classes of languages for which there exists a polynomial-time Turing machine $M$ whose acceptance probability on input $x$ is non-zero, at least 1/2, and not equal to 1/2, respectively, only when $x$ is in the language. 

Of these classes, only $\NP$ is known to be part of the polynomial hierarchy ($\PH$), which is a class composed of an infinite number of levels generalizing the notion of $\NP$. Furthermore, the other three classes, $\#\P$, $\PP$, and $\co\C_=\P$, which we refer to as counting classes, are known to be hard for the $\PH$: Toda's theorem \cite{toda1991pp} tells us that a $\#\P$ or $\PP$ oracle is sufficient to solve any problem in the $\PH$ in polynomial time, and other work by Toda and Ogiwara \cite{toda1992counting} shows that there is a randomized reduction from any problem in the $\PH$ to a $\co\C_=\P$ problem. Stated another way, if $\PP$ or $\co\C_=\P$ were to be contained in a level of the $\PH$, the $\PH$ would necessarily collapse, meaning that the entire $\PH$ would be contained within one of its levels. For example, if $\P=\NP$, then the entire $\PH$ would be equal to $\P$, its zeroth level. The assumption that the $\PH$ does not collapse is thus a stronger version of the statement $\P \neq \NP$, and it is widely believed for similar reasons.

Furthermore, these counting classes can be connected to quantum circuits. Aaronson showed that $\PP = \mathsf{PostBQP}$ \cite{aaronson2005quantum}, where $\mathsf{PostBQP}$ is the set of problems solvable by quantum circuits that have the (unphysical) power to choose, or \textit{postselect} the value of measurement outcomes that normally would be probabilistic. By contrast, classical circuits endowed with this same power form the class $\mathsf{PostBPP}$ which is known to lie in the third level of the $\PH$ \cite{han1997threshold}.

The story is similar for $\co\C_=\P$. It was shown that $\co\C_=\P = \NQP$ \cite{fenner1999determining}, where $\NQP$ is the quantum generalization of the class $\NP$, defined to be the set of languages $L$ for which there exists a polynomial-time uniformly generated\footnote{We say a circuit family
$\{C_x\}$
is \textit{uniformly generated} if there exists a polynomial-time Turing machine that
on input $x$ outputs the description of the circuit $\{C_x\}$.} family of circuits $\{C_x\}$ such that for all strings $x$, $x$ is the language $L$ if and only if the quantum circuit $C_x$ has a non-zero acceptance probability.  This can also be thought of as $\PostBQP$ with one-sided error.
If there existed an efficient classical algorithm to produce samples from the output distribution of quantum circuits up to constant multiplicative error, then $\NP$ would be equal to $\NQP$, and therefore to $\co\C_=\P$, leading to the collapse of the $\PH$ (to the second level \cite{fenner1999determining,fujii2014impossibility,fujii2015power}).

We refer to simulation algorithms of this type as approximate simulation algorithms with multiplicative error. Stated precisely, if $Q(y)$ is the probability that a quantum circuit produces the output $y$, then a classical simulation algorithm has multiplicative error $\epsilon$ if its probability of producing outcome $y$ is $P(y)$ and
\begin{equation}
\lvert P(y)-Q(y) \rvert \leq \epsilon Q(y)
\end{equation}
for all possible outcomes $y$.

This contrasts with a simulation algorithm with additive error $\epsilon$, for which 
\begin{equation}\label{eq:additiveerror}
\sum_y\lvert P(y)-Q(y) \rvert \leq \epsilon .
\end{equation}

The argument we have sketched only rules out polynomial-time simulation algorithms with multiplicative error. \hl{ In Section \ref{sec:adderror}, we discuss arguments \cite{aaronson2011computational, bremner2016average,bouland2017quantum, bouland2018quantum} that rule out additive-error simulation algorithms. These are more complex and require imposing additional conjectures.}

\subsection{IQP Circuits}
The previous argument only considers simulation algorithms for arbitrary quantum circuits, but the result can be extended to also rule out efficient simulation algorithms for subclasses of quantum circuits. An example of one such subclass is the set of instantaneous quantum circuits \cite{shepherd2009temporally,bremner2010classical}. Problems that can be solved by instantaneous quantum circuits with a polynomial number of gates form the Instantaneous Quantum Polynomial-time (IQP) complexity class, and we will refer to the circuits themselves as IQP circuits. There are several equivalent ways to define the IQP model; we do so as follows.

An IQP circuit is a circuit where a Hadamard gate is applied to each qubit at the beginning and end of the computation, but the rest of the gates, which we refer to as the \textit{internal} gates, are diagonal. Each qubit begins in the $\ket 0$ state but is immediately sent to the $\ket + = H\ket 0 = (\ket 0 + \ket 1)/\sqrt{2}$ state under the Hadamard operation, and each qubit is measured at the end of the computation in the computational basis. All of the internal diagonal gates commute, and therefore can be implemented in any order. An example of an IQP circuit is shown in Figure \ref{fig:IQPcircuitexample}.

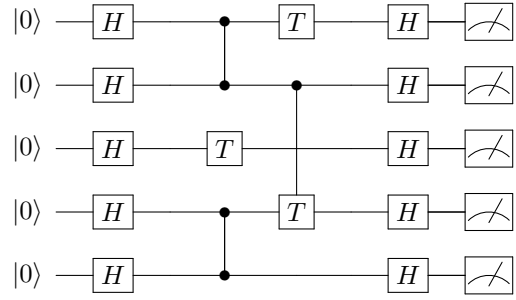
\begin{figure}[ht]
\centerline{
\Qcircuit @C=1.4em @R=1.0em {
\lstick{\ket 0} & \gate{H} & \qw & \ctrl{1} & \gate{T} & \qw & \gate{H} & \meter \qw \\
\lstick{\ket 0} & \gate{H} & \qw & \ctrl{-1} & \ctrl{2} & \qw & \gate{H} & \meter \qw \\
\lstick{\ket 0} & \gate{H} & \qw & \gate{T} & \qw & \qw & \gate{H} & \meter \qw\\
\lstick{\ket 0} & \gate{H} & \qw & \ctrl{1} & \gate{T} & \qw & \gate{H} & \meter \qw \\
\lstick{\ket 0} & \gate{H} & \qw & \ctrl{-1} & \qw & \qw & \gate{H} & \meter \qw 
}
}
\caption{\label{fig:IQPcircuitexample} Example of an IQP circuit. Each qubit must begin and end with a Hadamard gate, and all internal gates must be diagonal in the $Z$ basis. The vertical lines indicate controlled operations, and $T$ refers to the gate $T = \exp(-i\pi Z/8)$.}
\end{figure}

%%%%%%%%%%%%%%%%%%%%% SUBSECTION %%%%%%%%%%%%%%%%%%%%%%%%%%%%%%%%

\subsection{Quantum approximate optimization algorithm (QAOA) circuits}

Another class of circuits that is not efficiently simulable classically if the polynomial hierarchy does not collapse are quantum approximate optimization algorithm (QAOA) circuits \cite{farhi2014quantum, farhi2016quantum}, which have some similarities with IQP circuits. In a sense, QAOA can be thought of as multiple rounds of instantaneous operations.

A QAOA circuit operates on $n$ qubits, which begin in the $\ket 0$ state but are immediately hit with a Hadamard gate, as in the IQP model (note that \cite{farhi2014quantum, farhi2016quantum} chose a different but equivalent convention). An integer $p$, and angles $\gamma_i, \beta_i$ for $i = 1,2, \ldots, p$ are chosen. A diagonal Hamiltonian $C$ is specified such that $C = \sum_\alpha C_\alpha$ where each $C_\alpha$ is a constraint on a small subset of the bits, meaning for any bit string $z$, either $C_\alpha\ket z = 0$ or $C_\alpha \ket z = \ket z$ and only a few bits of $z$ are involved in determining which is the case. We define the Hamiltonian $B = \sum_{j=1}^n X_j$, where $X_j$ is the Pauli-$X$ operation applied to qubit $j$, and let $U(H,\theta) = \exp(-iH\theta)$. The remainder of the circuit consists of applying $U(C,\gamma_1)$, $U(B,\beta_1)$, $U(C,\gamma_2)$, $U(B,\beta_2)$, etc. for a total of $2p$ \hl{operations}. Finally the qubits are measured in the computational basis. The general framework for a QAOA circuit is depicted in Figure \ref{fig:QAOAcircuitexample}.

\begin{figure}[ht]
\centerline{
\Qcircuit @C=1.0em @R=1.0em {
\lstick{\ket 0} & \gate{H} & \multigate{4}{\rotatebox{270}{$U(C,\gamma_1)$}} & \multigate{4}{\rotatebox{270}{$U(B,\beta_1)$}} & \push{\;  \ldots   \;} \qw & \multigate{4}{\rotatebox{270}{$U(C,\gamma_p)$}} & \multigate{4}{\rotatebox{270}{$U(B,\beta_p)$}} & \meter \qw \\
\lstick{\ket 0} & \gate{H}  & \ghost{\rotatebox{270}{$U(C,\gamma_1)$}} & \ghost{\rotatebox{270}{$U(B,\beta_1)$}}  & \push{\;  \ldots   \;} \qw & \ghost{\rotatebox{270}{$U(C,\gamma_p)$}} & \ghost{\rotatebox{270}{$U(B,\beta_p)$}} & \meter \qw \\
\lstick{\ket 0} & \gate{H} & \ghost{\rotatebox{270}{$U(C,\gamma_1)$}} & \ghost{\rotatebox{270}{$U(B,\beta_1)$}} &  \push{\;  \ldots   \;} \qw & \ghost{\rotatebox{270}{$U(C,\gamma_p)$}} & \ghost{\rotatebox{270}{$U(B,\beta_p)$}} & \meter \qw\\
\lstick{\ket 0} & \gate{H} & \ghost{\rotatebox{270}{$U(C,\gamma_1)$}} & \ghost{\rotatebox{270}{$U(B,\beta_1)$}} &  \push{\;  \ldots   \;} \qw & \ghost{\rotatebox{270}{$U(C,\gamma_p)$}} & \ghost{\rotatebox{270}{$U(B,\beta_p)$}} & \meter \qw \\
\lstick{\ket 0} & \gate{H}  & \ghost{\rotatebox{270}{$U(C,\gamma_1)$}} & \ghost{\rotatebox{270}{$U(B,\beta_1)$}} &  \push{\;  \ldots   \;} \qw & \ghost{\rotatebox{270}{$U(C,\gamma_p)$}} & \ghost{\rotatebox{270}{$U(B,\beta_p)$}} & \meter \qw
}
}
\caption{\label{fig:QAOAcircuitexample} Framework for a QAOA circuit. Each qubit begins with a Hadamard gate, and then $2p$ \hl{operations} are performed alternating between applying Hamiltonian $C$ and applying Hamiltonian $B$.}
\end{figure}
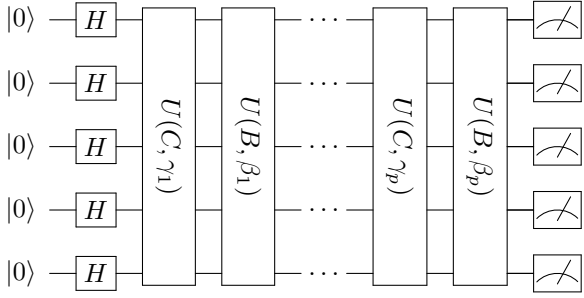

Since $U(C,\gamma_j) = \prod_\alpha U(C_\alpha,\gamma_j)$, the gate $U(C,\gamma_j)$ can be performed as a sequence of commuting gates that perform the unitaries associated with the constraints $C_\alpha$. Thus each $U(C,\gamma_j)$ could form the internal portion of an instantaneous quantum circuit.

Importantly, since the operator $C$ is a sum of many constraints, it represents a constraint satisfaction problem. For all bit strings $z$, $C\ket z = \lambda_z \ket z$, and a common problem asks us to find the maximum value of $\lambda_z$. There is evidence that QAOA circuits might be able to approximate this optimum value of $\lambda_z$ more efficiently than classical algorithms when $p > 1$ \cite{farhi2014quantum}, so in comparison to IQP circuits, QAOA circuits might have more practical value. 

\subsection{Boson sampling circuits}

The IQP and QAOA models are restrictions on the more general quantum circuit model. But quantum circuits are not the only model for computation on a quantum device. Linear quantum optical experiments, for example, can be modeled as a system of beam splitters and phase shifters acting upon identical photons existing in a certain number of different optical modes. Like the IQP and QAOA models, the linear optical model is not believed to be as powerful as the general quantum circuit model, but under the assumption that the $\PH$ does not collapse, it has been shown that classical simulation up to constant multiplicative error requires more than polynomial time \cite{aaronson2011computational}.

The basic framework \cite{aaronson2011computational} for the linear optical model is as follows. Suppose the system has $n$ photons among $m$ modes. A state of the system is a superposition $\sum_R \alpha_R \ket R$, where each $\ket R$ corresponds to a configuration of the $n$ photons among the $m$ modes, represented by the tuple $R=(r_1, \ldots, r_m)$ where each $r_i$ is a non-negative integer and $\sum_ir_i = n$. 

Passing these photons through a linear optical network composed of beam splitters and phase shifters, which we call a \textit{boson sampling circuit}, gives rise to a transformation on this Hilbert space. Valid transformations can be written as $\phi(U)$, where $U$ is any $m \times m$ unitary and $\phi$ is a fixed $\binom{n+m-1}{n}$-dimensional representation of $\mathcal{U}(m)$. The unitary $U$ fully describes the choice of circuit, and any $U$ can be exactly implemented using only $m(m+1)/2$ total beam splitters and phase shifters \cite{reck1994experimental}. We can define $\phi(U)$ by its matrix elements $\bra R \phi(U) \ket{R'}$, which will be related to the permanent of $n \times n$ matrices formed from $U$. The permanent of an $n \times n$ matrix $A$ is given by the formula

\begin{equation}\label{eq:permanent}
\text{Per}(A) = \sum_{\sigma \in \mathcal{S}_n}\prod_{i=1}^n A_{i,\sigma(i)},
\end{equation}
where $\mathcal{S}_n$ is the group of permutations on $\{1,\ldots, n\}$. Then, the matrix elements are

\begin{equation}
\bra R \phi(U)\ket{R'} = \frac{\text{Per}(U_{(R,R')})}{\sqrt{r_1!\ldots r_m!r'_1!\ldots r'_m!}},
\end{equation}
where $U_{(R,R')}$ is the $n \times n$ matrix formed by taking $r_i$ copies of row $i$ and $r'_j$ copies of column $j$ from $U$ \cite{aaronson2011computational}. As an example, if $n=3$, $m=2$, $R = (2,1)$, $R' = (1,2)$, and

\begin{equation}
U=\frac{1}{\sqrt{2}}
\begin{bmatrix}
    1     &  i\\
    -i      & -1
\end{bmatrix},
\end{equation}
then

\begin{equation}
U_{(R,R')}=\frac{1}{\sqrt{2}}
\begin{bmatrix}
    1     &  i & i\\
    1      & i & i \\
    -i	& 	-1	& -1
\end{bmatrix}.
\end{equation}

This sampling task is called \texttt{BosonSampling} since it could (in theory) be applied to any system of not only photons but any non-interacting bosons.

\section{Simulation algorithms with multiplicative error}\label{sec:multerror}

The goal of this section will be to perform a fine-grained analysis that rules out certain multiplicative-error classical simulation algorithms for the three previously mentioned families of quantum circuits.

\subsection{Degree-3 polynomials and the problem \texttt{poly3-NONBALANCED}} \label{sec:deg3polys}

Each of the three quantum circuit families that we have defined are especially amenable to this analysis due to their natural connection to \textit{specific} counting problems. 

The specific counting problem we will use for our analysis of IQP and QAOA we call \texttt{poly3-NONBALANCED}. The input to the problem is a polynomial over the field $\mathbb{F}_2$ in $n$ variables with degree at most 3 and no constant term. Since the only non-zero element in $\mathbb{F}_2$ is 1, every term in the polynomial has coefficient 1. One example could be $f(z) = z_1 + z_2 + z_1z_2 + z_1z_2z_3$. Evaluating $f$ for a given string $z$ to determine whether $f(z) = 0$ or $f(z) = 1$ can be done efficiently, but since there are $2^n$ possible strings $z$, the brute-force method takes exponential time to count the number of strings $z$ for which $f(z) = 0$, or equivalently, to compute $\gap(f)$ where $\gap$ is given by Eq.~\eqref{eq:gapdef}. LPTWY~\cite{lokshtanov2017beating} gave a deterministic algorithm for computing the gap of degree-3 polynomials in time scaling slightly better than brute force, but it still has exponential time --- $\text{poly}(n)2^{0.9965n}$. 

The question posed by \texttt{poly3-NONBALANCED} is whether $\gap(f) \neq 0$, that is, whether $f$ has the same number of 0 and 1 outputs. Thus, \texttt{poly3-NONBALANCED} is in the class $\co\C_=\P$. 

The problem \texttt{poly3-NONBALANCED} is a natural problem to work with because there is an elegant correspondence between degree-3 polynomials and IQP circuits involving Pauli $Z$ gates, controlled-$Z$ ($CZ$) gates, and controlled-controlled-$Z$ ($CCZ$) gates \cite{montanaro2017quantum}. Specifically, if $f$ is degree 3 then let
\be U'_f = \sum_{z\in \bbF_2^n} (-1)^{f(z)} \ket z \bra z
\ee
and let $U_f = H^{\otimes n} U'_f H^{\otimes n}$.
We can implement an IQP circuit $C_f$ that evaluates to $U_f$ as follows: if the term $z_i$ appears in $f$, then within the diagonal portion of $C_f$ we perform the gate $Z$ on qubit $i$; if the term $z_iz_j$ appears, we perform the $CZ$ gate between qubits $i$ and $j$; and if the term $z_iz_jz_k$ appears, we perform the $CCZ$ gate between the three qubits. For example, for the polynomial $f(z) = z_1 + z_2 + z_1z_2 + z_1z_2z_3$, the circuit $\mathcal{C}_f$ is  shown in Figure \ref{fig:polynomialcircuitexample}.

\begin{figure}[ht]
\centerline{
\Qcircuit @C=1.4em @R=1.0em {
\lstick{\ket 0} & \gate{H} & \gate{Z} & \qw & \ctrl{1} & \ctrl{2} & \gate{H} & \qw \\
\lstick{\ket 0} & \gate{H} & \qw & \gate{Z} & \ctrl{-1} & \ctrl{-1} & \gate{H} &  \qw \\
\lstick{\ket 0} & \gate{H} & \qw & \qw & \qw & \ctrl{-1} & \gate{H} & \qw
}
}
\caption{\label{fig:polynomialcircuitexample} IQP circuit $\mathcal{C}_f$ corresponding to the degree-3 polynomial $f(z) = z_1 + z_2 + z_1z_2 + z_1z_2z_3$. The unitary $U_f$ implemented by the circuit has the property that $\bra{\bar{0}} U_f \ket{\bar{0}} = \gap(f)/{2^n}$ where in this case $n = 3$. }
\end{figure}
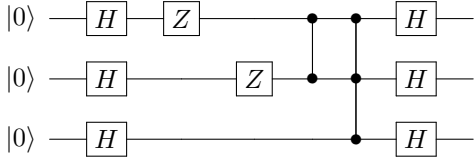

The crucial property of this correspondence is that $\bra{\bar 0} U_f \ket{\bar 0} = \frac{\gap(f)}{2^n}$, where $\ket{\bar 0}$ is shorthand for the starting $\ket 0 ^{\otimes n}$ state. This is easily seen by noting that the initial set of $H$ gates generates the equal superposition state $\ket B =\sum_{x=0}^{2^n-1} \ket x/\sqrt{2^n}$, so $\bra{\bar 0} U_f \ket{\bar 0} = \bra B U'_f \ket B$ where $U'_f$ is implemented by the internal diagonal portion of $\mathcal{C}_f$. Since $U'_f$ applies a $(-1)$ phase to states $\ket z$ for which $f(z)=1$, $\bra{\bar 0} U_f \ket{\bar 0} = \sum_{y=0}^{2^n-1}\sum_{x=0}^{2^n-1} (-1)^{f(x)}\braket{y}{x}/2^n = \sum_{x=0}^{2^n-1} (-1)^{f(x)}/2^n = \gap(f)/2^n$. Thus, $\gap(f)$ can be computed by calculating the amplitude of the $\ket{\bar 0}$ state produced by the circuit. If we define acceptance to occur when $\ket{\bar 0}$ is obtained upon measurement, then the circuit $C_f$ has non-zero acceptance probability only when $\gap(f) \neq 0$. This illustrates an explicit $\NQP$ algorithm for \texttt{poly3-NONBALANCED}, which was guaranteed to exist since $\NQP = \co\C_=\P$. 

Also crucial to note is that \texttt{poly3-NONBALANCED} is complete for the class $\co\C_=\P$. This is shown by adapting Montanaro's proof \cite{montanaro2017quantum} that computing $\gap(f)$ for a degree-3 polynomial $f$ over $\mathbb{F}_2$ is $\#\P$-complete. In that proof, Montanaro reduces from the problem of computing $\gap(g)$ for an arbitrary Boolean function $g$, which is $\#\P$-complete by definition. Since whether $\gap(g) \neq 0$ is $\co\C_=\P$-complete by definition, and the reduction has $\gap(g) \neq 0$ if and only if $\gap(f) \neq 0$, this also shows that \texttt{poly3-NONBALANCED} is $\co\C_=\P$-complete. One immediate consequence of this fact is that $\mathsf{NIQP}$, the class $\NQP$ restricted to quantum circuits of the IQP type, is equal to $\co\C_=\P$ (and hence $\NQP$), since the circuit $C_f$ is an $\mathsf{NIQP}$ solution to a $\co\C_=\P$-complete problem. 

\subsection{The permanent and the problem \texttt{per-int-NONZERO}}

In close analogy to the correspondence between degree-3 polynomials and IQP circuits composed of $Z$, $CZ$, and $CCZ$ gates, there is a correspondence between matrix permanents and boson sampling circuits.

We have already seen in the definition of the linear optical model that any amplitude in a boson sampling circuit on $n$ photons can be recast as the permanent of an $n \times n$ matrix, but the converse is also true: the permanent of any $n \times n$ matrix can be encoded into the amplitude of a boson sampling circuit on $n$ photons, up to a known constant of proportionality. 

To see how this works, given an $n \times n$ complex matrix $A$, we will construct a $2n \times 2n$ unitary matrix $U_A$ whose upper-left $n\times n$ block is equal to $cA$ for some $c>0$.  If we take $R=R'=(1^n,0^n)$ (i.e.~1 repeated $n$ times, followed by 0 repeated $n$ times), then we will have $\Per(U_{A(R,R')}) = c^n \Per(A)$.  Thus $\Per(A)$ is proportional to a particular boson sampling amplitude with $c$ an easily computable proportionality constant.

We can choose $c$ to be $\leq \|A\|^{-1}$, where $\|A\|$ is the largest singular value of $A$.  (Note that if we want the proportionality to hold uniformly across some class of $A$, we should choose $c$ to satisfy $c\|A\| \leq 1$ for all $A$ in this class.)
Then $\{cA, \sqrt{I_n - c^2 A^\dag A}\}$ are Kraus operators for a valid quantum operation, where $I_n$ is the $n \times n$ identity matrix, and
\be
\begin{bmatrix}
cA \\  \sqrt{I_n - c^2 A^\dag A} 
\end{bmatrix}
\ee
is an isometry.  We can extend this isometry to the following unitary.
\begin{equation}
U_A = 
\begin{bmatrix}
    cA     &  D \\
    \sqrt{I_n - c^2 A^\dag A}      & -\frac{1}{\sqrt{I_n-c^2A^\dag A}} c A^\dag D
\end{bmatrix},
\end{equation}
where $D = \left(I_n+c^2A(I_n-c^2 A^\dag A)^{-1}A^\dag\right)^{-1/2}$,
which is well defined since the argument of the inverse square root is positive definite and Hermitian. Thus the permanent of an arbitrary $n \times n$ matrix can be encoded into a boson sampling circuit with $n$ photons and $2n$ modes.

The matrix permanent is playing the role for boson sampling circuits that the gap of degree-3 polynomials played for IQP circuits with $Z$, $CZ$, and $CCZ$ gates; thus, it is natural to use the computational problem of determining if the permanent of an integer-valued matrix is not equal to 0, which we call \texttt{per-int-NONZERO}, in place of \texttt{poly3-NONBALANCED}.

In fact, \texttt{per-int-NONZERO} and \texttt{poly3-NONBALANCED} have several similarities. For example, like computing the number of zeros of a degree-3 polynomial, computing the permanent of an integer-valued matrix is $\#\P$-complete, a fact famously first demonstrated by Valiant \cite{valiant1979complexity}, and later reproved by Aaronson \cite{aaronson2011linear} using the linear optical framework. This completeness extends to \texttt{per-int-NONZERO}, which we show in Appendix \ref{app:coC=Phard} is $\co\C_=\P$-complete by reduction from $\texttt{poly3-NONBALANCED}$.

Additionally, for both problems, the best known algorithm is exponential and has runtime close to or equaling $2^n$. While \texttt{poly3-NONBALANCED} can be solved in poly$(n)2^{0.9965n}$ time, the best known algorithm for computing the permanent \cite{bjorklund2018generalized} requires $2^{n-\Omega(\sqrt{n/\log\log(n)})}$ deterministic time, which is only a subexponential improvement over the naive algorithm that utilizes Ryser's formula for the permanent \cite{ryser1963combinatorial} and requires at least $n2^n$ basic arithmetic operations. Using Ryser's formula is an improvement over the $O(n!)$ time steps implied by Eq.~\eqref{eq:permanent}, but its scaling is reminiscent of that required to solve a $\#\P$ problem by brute force. In principle it is possible that a faster algorithm exists for $\texttt{per-int-NONZERO}$, where we do not care about the actual value of the permanent, only whether it is nonzero, but  such methods are only known in special cases, such as nonnegative matrices.

Crucially, our construction shows that boson sampling circuits can solve \texttt{per-int-NONZERO} in non-deterministic polynomial time, since given $A$ we have shown how to construct a circuit corresponding to unitary $U_A$ with acceptance probability that is non-zero only when $\Per(A)$ is non-zero. This shows that $\mathsf{NBosonP}$, the linear optical analogue of $\mathsf{NIQP}$, is equal to $\co\C_=\P$ and by extension, to $\NQP$. 

\subsection{Lower bounds for multiplicative-error simulations}

\subsubsection{For IQP Circuits}

In the previous section, we described how to construct an $n$-qubit IQP circuit $C_f$ corresponding to a degree-3 polynomial $f$ over $n$ variables such that the acceptance probability of $C_f$ is non-zero if and only if $\gap(f)\neq 0$. The number of terms in $f$, and hence the number of internal diagonal gates in $C_f$ is at most
\begin{align}
g_1(n) &= \binom{n}{3}+\binom{n}{2}+\binom{n}{1} \nonumber\\
&= (n^3+ 5n)/6.
\end{align}
Now, suppose we had a classical algorithm that, for any $q$, produces samples from the output distribution of any IQP circuit with $q$ qubits and $g_1(q)$ internal gates, up to some multiplicative error constant, in $s_1(q)$ time steps for some function $s_1$. Throughout, we will assume all classical algorithms run in the Word RAM model of computation. 

Using this algorithm to simulate the IQP circuit $C_f$ generates a non-deterministic classical algorithm for $\texttt{poly3-NONBALANCED}$ running in $s_1(n)$ time steps. That is, the classical probabilistic algorithm that results from this simulation accepts on at least one computational path if only if the function $f$ is not balanced. 

Now, we impose a fine-grained version of the non-collapse assumption, which we motivate later in the section.

\begin{conj}\label{conj:poly3NSETH}[poly3-NSETH($\ptNc $)]

Any non-deterministic classical algorithm (in the Word RAM model of computation) that solves \texttt{poly3-NONBALANCED} requires in the worst case $2^{\ptNc n-1}$ time steps, where $n$ is the number of variables in the \texttt{poly3-NONBALANCED} instance. 
\end{conj}

In the Word RAM model with word size $w$, memory is infinite and basic arithmetic operations on words of length $w$ take one time step. For concreteness, we assume that $w = \log_2(N)$ where $N$ is the length of the input encoding the degree-3 polynomial ($N = O(g_1(n)\log_2(n))$). This way the words can index the locations where the input data is stored. The Word RAM model has previously been used for fine-grained analyses \cite{williams2015hardness} and aims to represent how a real computer operates as faithfully as possible. 

Our conjecture immediately yields a lower bound on the simulation function $s_1$:
\begin{equation}\label{eq:iqpbound}
s_1(n) \geq 2^{\ptNc n-1}.
\end{equation}

This lower bound result relies on poly3-NSETH($\ptNc $), which we have not yet motivated. In particular, for our qubit calculations we will take the specific value of $\ptNc  = 1/2$.  This  value is comfortably below the best known limit $\ptNc  < 0.9965$ from \cite{lokshtanov2017beating}, whose algorithm is reproduced in Appendix \ref{app:algo}. In Section \ref{sec:evidence}, we attempt to provide additional motivation for poly3-NSETH(1/2) by showing its consistency with other fine-grained conjectures.

To our knowledge, the best known upper bound on $s_1(n)$ comes from the the naive $\text{poly}(n)2^n$-time ``Schr\"{o}dinger-style'' simulation algorithm that updates each of the $2^n$ amplitudes describing the state vector after each gate is performed, so this lower bound is not tight. 

\subsubsection{For QAOA circuits}

To perform the same analysis for QAOA circuits, we will turn the IQP circuit $C_f$ into a QAOA circuit. The modifications required are straightforward. We set $p$, the number of rounds of QAOA computation, equal to 1, \hl{and we let rotation angles $\gamma = \pi/2$ and $\beta = \pi/4$}. The first layer of Hadamard gates in $C_f$ is already built into the QAOA framework. To implement the $Z$, $CZ$, and $CCZ$ gates we write $Z = \exp(-i 2\gamma \ket 1 \bra 1)$, $CZ = \exp(-i 2\gamma \ket{11} \bra{11})$, and $CCZ = \exp(-i 2\gamma \ket{111} \bra{111})$ and build our constraint Hamiltonian $C$ accordingly: for each $Z$ gate we add two copies of the constraint that is satisfied only when the bit acted upon is 1; for each $CZ$ gate we add two copies of the constraint that is satisfied when both bits involved are 1; and for each $CCZ$ gate we add two copies of the constraint that is satisfied when all three bits involved are 1. Now, the operation $\exp(-i\gamma C)$ has exactly the effect of all the $Z$, $CZ$, and $CCZ$ gates combined. 

The final step is to implement the final column of $H$ gates, which is not built into the QAOA framework. First we define

\begin{equation}
\tilde{H}= \frac{1}{\sqrt{2}}
\begin{bmatrix}
    1     &  -i \\
    -i      &  1 
\end{bmatrix} = \exp\left(-i \frac{\pi}{4} X\right) = \exp(-i \beta X).
\end{equation}
Note that \hl{$\tilde{H} = H\exp(-i \frac{\pi}{4} Z)H$}, which can be rewritten \hl{$H = \exp(-i \gamma \ket 1\bra 1)H\tilde{H}$ (up to an unimportant global phase)}. \hl{Thus, we can replace the $H$ gates in the final column of the circuit $C_f$ with right-hand-side of the previous expression, the first factor of which can be performed by adding one copy of the $\ket 1 \bra 1$ constraint to $C$}. As described in \cite{farhi2016quantum}, the second factor ($H$ gate) can be implemented by introducing an ancilla qubit and \hl{four} new constraints between the original qubit and the ancilla. The original qubit is measured and if outcome $\ket 0$ is obtained, the state of the ancilla is $H$ applied to the input state on the original qubit. Thus we have teleported the $H$ gate onto the ancilla qubit within the QAOA framework. This is described in full in \cite{farhi2016quantum}, and we reproduce the gadget in Figure \ref{fig:qaoagadget}.

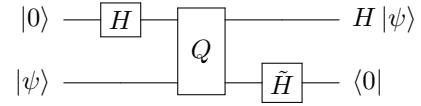
\begin{figure}[ht]
\centerline{
\Qcircuit @C=1.4em @R=1.0em {
\lstick{\ket 0} & \gate{H} & \multigate{1}{Q} & \qw  & \rstick{H\ket{\psi}} \qw \\
\lstick{\ket \psi} & \qw & \ghost{Q} & \gate{\tilde{H}} & \rstick{\bra{0}} \qw
}
}
\caption{\label{fig:qaoagadget} Gadget that uses an ancilla qubit to implement the $H$ gate within the QAOA framework. Here the gate $Q$ is the diagonal two-qubit gate \hl{$\text{diag}(1,i,1,-i)$}  which can be written as \hl{$\exp(-i\frac{\pi}{2}(3\ket{01}\bra{01}+\ket{11}\bra{11}))$}. Thus, it can be implemented by adding 4 constraints to the constraint Hamiltonian C. The $\tilde{H}$ gate is implemented by applying the Hamiltonian $B$ with $\beta=\pi/4$.}
\end{figure}

After replacing each $H$ gate with the gadget from Figure \ref{fig:qaoagadget}, every qubit begins with an $H$ gate, is acted upon by $\exp(-i\gamma C)$, and ends with a $\tilde{H}$ gate, which is implemented by the $\exp(-i\beta B)$ step of the QAOA framework. Thus, the resulting circuit is a QAOA circuit.

We had to introduce one ancilla per qubit in $C_f$, so our QAOA circuit has $2n$ qubits, instead of just $n$. However, it is still true that $\bra{\bar 0}V_f \ket{\bar 0} \propto \gap(f)$, where $V_f$ is now the unitary implemented by this new QAOA circuit and $\ket{\bar 0}$ is the state $\ket 0^{\otimes 2n}$. Hence the acceptance probability is non-zero if and only if $f$ is not balanced.

The circuit requires 2 constraints per term in the polynomial $f$, and an additional 5 constraints per qubit for the Hadamard gates at the end of the computation (1 from introducing $\tilde{H}$ and 4 from the gadget in Figure \ref{fig:qaoagadget}). This yields at most 
\begin{align}
g_2(\hl{2n})&=
2 g_1(n) + 5n \nonumber\\
&=
(n^3+20n)/ 3
\end{align}
constraints. 

As in the IQP case, we suppose a classical simulation algorithm produces samples from the output distribution of QAOA circuits with $q$ qubits and \hl{$g_2(q)$} constraints, up to multiplicative error constant, in time $s_2(q)$. Then, under the same conjecture poly3-NSETH($\ptNc $), we have
\begin{equation}
s_2(2n) \geq 2^{\ptNc n-1},
\end{equation}
which simplifies to
\begin{equation}\label{eq:qaoabound}
s_2(n) \geq 2^{\frac{\ptNc n}{2}-1}.
\end{equation}

The exponentiality of this lower bound is weaker by a factor of two in comparison to the lower bound for IQP circuits in Eq.~\eqref{eq:iqpbound}, due to the fact that one ancilla was introduced per variable to turn the circuit $C_f$ into a QAOA circuit. However, the best known upper bound for QAOA simulation is the naive $\text{poly}(n)2^n$ brute-force algorithm, as was the case for IQP circuits. This indicates that one might be able to eliminate the factor of two by replacing \texttt{poly3-NONBALANCED} with another problem.  Such a problem should be solvable by a QAOA circuit but not by non-deterministic algorithms running much faster than brute force.  We leave this for future work.

\subsubsection{For boson sampling circuits}

The story for boson sampling circuits is nearly identical, except using a conjecture related to the problem \texttt{per-int-NONZERO} instead of \texttt{poly3-NONBALANCED}. 

Given an integer-valued $n \times n$ matrix $A$, we showed previously how to construct a boson sampling circuit with $n$ photons, described by unitary $U_A$, that has non-zero acceptance probability only when $\Per(A) \neq 0$. This circuit has $2n$ modes, and hence requires at most
\begin{equation}
g_3(n) \hl{= (2n)(2n+1)/2} = 2n^2+n
\end{equation}
circuit elements, i.e.~beam splitters and phase shifters.

Paralleling our IQP and QAOA analysis, we suppose we have a classical algorithm that produces samples from the output distribution of a boson sampling circuit with $q$ photons and $g_3(q)$ total beam splitters and phase shifters, up to some multiplicative error constant, in $s_3(q)$ time steps for some function $s_3$. 

Using this algorithm to simulate the boson sampling circuit described by $U_A$ generates a non-deterministic algorithm for $\texttt{per-int-NONZERO}$ running in $s_3(n)$ time steps.

We replace poly3-NSETH($\ptNc $) with the version for \texttt{per-int-NONZERO}

\begin{conj}\label{conj:perintNSETH}[per-int-NSETH($\piNc $)]
Any non-deterministic classical algorithm (in the Word RAM model of computation) that solves \texttt{per-int-NONZERO} requires in the worst case $2^{\piNc n-1}$ time steps, where $n$ is the number of rows in the \texttt{per-int-NONZERO} instance.
\end{conj}

Unlike poly3-NSETH($\ptNc $), as far as we are aware there is no known better-than-brute force algorithm ruling out the conjecture for any value $\piNc < 1$. The algorithm in \cite{bjorklund2018generalized}, which is better-than-brute-force by subexponential factors rules out $\piNc = 1$.

This conjecture implies a lower bound on the simulation function

\begin{equation}\label{eq:bsbound}
s_3(n) \geq 2^{\piNc n-1}.
\end{equation}

Producing samples from the output of boson sampling circuits naively requires one to compute the permanent for many of the amplitudes. However, in the case of a binary output, where acceptance is defined to correspond to exactly one photon configuration, only one permanent need be calculated --- the one associated with the accepting configuration. Thus the asymptotic scaling of this lower bound when $\piNc =1-\delta$ is essentially tight with naive simulation methods as $\delta \rightarrow 0$, since Ryser's formula can be used to evaluate the permanent and simulate a boson sampling circuit in $O(n2^n)$ time steps.

\subsection{Evidence for conjectures}\label{sec:evidence}

Where previous quantum computational supremacy arguments only ruled out simulation algorithms with polynomial runtime, our analysis also rules out some algorithms with exponential runtime. These conclusions come at the expense of imposing stronger, fine-grained conjectures, but such assumptions are necessary for extracting the fine-grained lower bounds we seek.

Thus, our conjectures are necessarily less plausible than the statement that the $\PH$ does not collapse, and definitively proving our conjectures is impossible without simultaneously settling major open problems in complexity theory. However, we can give evidence for these conjectures by thinking about how one might try to refute them, and showing how they fit into the landscape of previously proposed fine-grained conjectures.

We start with poly3-NSETH($\ptNc $) and discuss why certain techniques for refuting it cannot work, how current techniques fall short of refuting it for values of $\ptNc $ significantly lower than 1, and why we should expect that completely different techniques would be needed to produce algorithms that rule out $\ptNc  < 1/2$. Then, we discuss how poly3-NSETH($\ptNc $) fits in consistently with other results in fine-grained complexity theory. Finally, we discuss how per-int-NSETH($\piNc $) is similar and different in these regards.

The conjecture poly3-NSETH($\ptNc $) asserts that determining whether a Boolean function is balanced takes non-deterministic exponential time, where that Boolean function takes the form of a degree-3 polynomial. It is worth noting that we can prove this conjecture with $\ptNc =1$ for Boolean functions in the black-box setting, where the non-deterministic algorithm can only interact with the Boolean function by querying its value on certain inputs. 

\begin{theorem}\label{thm:oracle-LB}
Let $f: \{0,1\}^n \rightarrow \{0,1\}$ be a Boolean function. A non-deterministic algorithm with black-box access to $f$ that accepts iff $\lvert\{x: f(x)=0\}\rvert \neq 2^{n-1}$, that is, iff $f$ is not balanced, must make at least $2^{n-1}+1$ queries to $f$. Moreover, this bound is optimal.
\end{theorem}

\begin{proof}
First we prove the lower bound on the number of queries. Suppose $M$ is a non-deterministic algorithm with black-box access to $f$ that accepts whenever $f$ is not balanced. Let $f_0$ be a Boolean function that is not balanced; thus, at least one computation path of $M$ accepts if $f=f_0$. Choose one such path and let $S \subset \{0,1\}^n$ be the set of queries made by $M$ on this computation path. Suppose for contradiction that $\lvert S \rvert \leq 2^{n-1}$. Since at most half the possible inputs are in $S$, it is possible to construct another Boolean function $f_1$ that is balanced and agrees with $f_0$ on the set $S$. Since $f_0$ and $f_1$ agree on $S$, the computation that accepted when $f=f_0$ will proceed identically and accept when $f=f_1$. Thus $M$ accepts when $f=f_1$, which is balanced, yielding a contradiction. We conclude that $\lvert S \rvert \geq 2^{n-1}+1$. 

We can see that it is possible for $M$ to achieve this bound as follows: $M$ non-deterministically chooses $2^{n-1}+1$ of the $2^n$ possible inputs to $f$ and queries $f$ on these inputs. If all of the queries yield the same value, it accepts. Otherwise, it rejects. If $f$ is balanced, $M$ will reject no matter which set of queries it makes, whereas if $f$ is not balanced, there is at least one set of $2^{n-1}+1$ inputs on which $f$ takes the same value and $M$ will accept, so the algorithm succeeds.
\end{proof}

Theorem \ref{thm:oracle-LB} shows that the poly3-NSETH(1) conjecture cannot be disproved using an algorithm that simply evaluates the degree-3 polynomial $f$ for different inputs. Indeed, the algorithm by LPTWY~\cite{lokshtanov2017beating} exploits the fact that the Boolean functions are degree-3 polynomials in order to refute poly3-NSETH($\ptNc $) for $\ptNc  > 0.9965$. Refuting poly3-NSETH($\ptNc $) for even smaller values of $\ptNc $ would require more techniques that further utilize the structure associated with the \texttt{poly3-NONBALANCED} problem.  

In fact, the algorithm in \cite{lokshtanov2017beating} is substantially more general than what is necessary for our purposes; their deterministic algorithm counts the number of solutions to a system of $m$ degree-$k$ polynomial equations over finite field $\mathbb{F}_q$. The problem \texttt{poly3-NONBALANCED} is concerned only with the case where $m=1$, $k=3$, $q=2$, and all that matters is whether the number of zeros is equal to half the possible inputs. For this special case, the algorithm is considerably simpler, and we reproduce it in Appendix \ref{app:algo}. The basic technique is as follows: we fix some fraction $(1-\delta)$ of the $n$ variables and call $R$ the number of zeros of $f$ consistent with those fixed values. We can compute in time $O(2^{0.15n+0.85\delta n})$ a representation of $R$ as a polynomial with integer coefficients over the $(1-\delta)n$ fixed variables. Then, (even though $R$ has an exponential number of monomials in its representation) it is noted that one can evaluate $R$ for all $2^{(1-\delta)n}$ possible inputs in total time $O(2^{(1-\delta)n})$, as long as $\delta < 0.0035$. By evaluating and summing $R$ on all of its inputs, we compute the total number of zeros, and the total runtime is $O(2^{(1-\delta)n})$, which is better than brute force when we choose $\delta$ positive.

Note that this algorithm is deterministic, and giving it the power of non-determinism can only make it faster. However, by inspection of the algorithm from \cite{lokshtanov2017beating}, we see no clear way for non-determinism to be directly utilized to further accelerate the algorithm. This is consistent with the finding in Theorem~\ref{thm:oracle-LB} that (asymptotically speaking) the best non-deterministic algorithms are no faster than the best deterministic algorithms for the \texttt{NONBALANCED} problem in the black-box setting. However, it is worth mentioning that a gap between best known deterministic and non-deterministic algorithms has been observed for certain \co\NP-hard problems, for example in \cite{woods1998unsatisfiable}, where the problem of determining the \textit{unsatisfiability} of a system of $m$ degree-2 polynomials in $n$ variables over $\mathbb  F_2$ is shown to be possible in $\text{poly}(n)2^{n/2}$ non-deterministic time, an improvement over best known $O(2^{0.8765n})$ deterministic solution from LPTWY \cite{lokshtanov2017beating}. Additionally, when randomness is added to non-determinism yielding what is known as a Merlin-Arthur (\MA) protocol, similar results can be shown even for \#\P-hard problems like computing the permanent of $n \times n$ matrices, which is possible in $\text{poly}(n)2^{n/2}$ $\MA$ time \cite{williams2016strong} compared to $O(2^n)$ deterministic time. These results cast some doubt on the assumption that non-determinism is a useless resource for solving \texttt{poly3-NONBALANCED} or \texttt{per-int-NONZERO}, although notably none of these methods appear to break the $\Omega(2^{n/2})$ barrier.

Additionally, we mention that the authors of LPTWY \cite{lokshtanov2017beating} were concerned primarily with showing that better-than-brute-force algorithms were possible, perhaps leaving room for optimization of their constants. In our reproduction of their algorithm when $m=1$, $k=3$, and $q=2$ in Appendix \ref{app:algo}, we have followed their analysis and optimized the constants where possible yielding a slightly better runtime than what is stated explicitly in their paper. 

The conclusion is that techniques exist to rule out poly3-NSETH(1) but not for values of $\ptNc $ much lower than 1, even after some attempt at optimization. Moreover, we now provide evidence that drastically different techniques would need to be used if one wished to rule out poly3-NSETH($1/2$); that is, ruling out poly3-NSETH($\ptNc $) when $\ptNc  < 1/2$ could not be done by making only slight modifications or improvements using the same approach from \cite{lokshtanov2017beating}. Our reasoning stems from the tradeoff between the two contributions to the runtime of the algorithm: first, the computation of the polynomial representation for $R$; and second, the evaluation of $R$ for all $2^{(1-\delta)n}$ possible inputs. When $\delta$ is smaller than $0.0035$, the second contribution dominates for a total runtime $O(2^{(1-\delta)n})$. However, if this step were to be improved to allow for $\delta$ to exceed $1/2$, the first contribution to the runtime would begin to dominate for a total runtime of $O(2^{0.15n+0.85\delta n}) > O(2^{n/2})$. In other words, if we try to fix fewer than half of the variables, computing the representation of $R$ (which involves cycling through the $2^{\delta n}$ strings of unfixed variables) will necessarily take longer than evaluating $R$ and ultimately it will be impossible to produce an algorithm with runtime below $2^{n/2}$ through this method. While this is no proof, it increases the plausibility of poly3-NSETH(1/2).

Next we discuss how poly3-NSETH($\ptNc $) contrasts with previously proposed fine-grained conjectures. Well-known conjectures include the Exponential Time Hypothesis (ETH), which claims that there exists some $c$ such that no $O(2^{cn})$ time algorithm for $k$-\texttt{SAT} exists, and the Strong Exponential-Time Hypothesis (SETH) \cite{impagliazzo2001problems,calabro2009complexity}, which states that for any $\epsilon$ one can choose $k$ large enough such that there is no $O(2^{(1-\epsilon)n})$ algorithm for $k$-\texttt{SAT}. In other words, SETH states that no algorithm for $k$-\texttt{SAT} does substantially better than the naive brute-force algorithm when $k$ is unbounded. 

There is substantial evidence for ETH and SETH, even beyond the fact that decades of research on the \texttt{SAT} problem have failed to refute them. For instance, SETH implies fine-grained lower bounds on problems in $\P$ that match long-established upper bounds. One example is the orthogonal vectors (\texttt{OV}) problem, which asks if a set of $n$ vectors has a pair that is orthogonal. There is a brute-force $O(n^2)$ solution to \texttt{OV}, but $O(n^{2-\epsilon})$ is impossible for any $\epsilon > 0$ assuming SETH \cite{williams2005new,williams2014finding}. Thus, SETH being true would provide a satisfying rationale for why attempts to find faster algorithms for problems like \texttt{OV} have failed. On the other hand, the refutation of SETH would imply the existence of novel circuit lower bounds \cite{jahanjou2015local}. 

% Our choice of $\ptNc =1/2$ is bolstered by evidence that poly3-NSETH($\ptNc $) should hold for \textit{some} constant $\ptNc $, even if the evidence for a specific value of $\ptNc $ is less concrete. First of all, for hard problems that admit better-than-brute-force algorithms, it is common for the runtime of best known algorithms to plateau at a certain value. For example ...

There are yet more fine-grained conjectures: replacing the problem $k$-\texttt{SAT} with $\#k$-\texttt{SAT} yields \#ETH and \#SETH, the counting versions of ETH and SETH. These hypotheses have interesting consequences of their own; for example, \#ETH implies that computing the permanent cannot be done in subexponential time \cite{dell2014exponential}. Additionally, if $k$-\texttt{TAUT} is the question of whether a $k$-DNF formula is satisfied by \textit{all} its inputs (which is $\co\NP$-complete), then the statement that no $O(2^{(1-\epsilon)n})$ algorithm exists for $k$-\texttt{TAUT} with unbounded $k$ is called the Non-deterministic Strong Exponential Time Hypothesis (NSETH) \cite{carmosino2016nondeterministic}. Like SETH, NSETH's refutation would imply circuit lower bounds \cite{jahanjou2015local,carmosino2016nondeterministic}. Additionally, NSETH is consistent with unconditional lower bounds that have been established in proof complexity \cite{beck2013strong,pudlak1999lower}.

The conjecture poly3-NSETH($\ptNc $) is similar to NSETH in that it asserts the non-existence of non-deterministic algorithms for a problem that is hard for $\co\NP$ (indeed, \texttt{poly3-NONBALANCED} is hard for the whole \PH), and it is similar to \#SETH in that it considers a counting problem. It is different from all of these conjectures because it is not based on satisfiability formulas, but rather on degree-3 polynomials over the field $\mathbb{F}_2$, a problem that has been far less studied. Additionally, poly3-NSETH($\ptNc $) goes beyond previous conjectures to assert not only that algorithms require $O(2^{\ptNc n})$ time, but that they actually require at least $2^{\ptNc n-1}$ time steps. It is not conventional to worry about constant prefactors as we have in this analysis, but doing so is necessary to perform practical runtime estimates. On this front, our analysis is robust in the sense that if poly3-NSETH($\ptNc $) or per-int-NSETH($\piNc $) were to fail by only a constant prefactor, the number of additional qubits we would estimate would increase only logarithmically in that constant. 

We are unable to show that poly3-NSETH($\ptNc $) is formally implied by any of the previously introduced conjectures. However, assuming ETH, we can prove that the deterministic version of poly3-NSETH($\ptNc $) holds for at least some $\ptNc $, i.e.~that there does not exist a deterministic $O(2^{\ptNc n})$ time algorithm for \texttt{poly3-NONBALANCED}.

\begin{theorem}
Assuming ETH, there exists a constant $\ptNc $ such that every deterministic algorithm that solves \texttt{poly3-NONBALANCED} requires $\Omega(2^{\ptNc n})$ time.
\end{theorem}

\begin{proof}
Suppose for contradiction that no such constant existed; thus for any $\ptNc $ there is an algorithm for \texttt{poly3-NONBALANCED} running in less than $O(2^{\ptNc  n})$ time. We give a reduction from $k$-\texttt{SAT} to \texttt{poly3-NONBALANCED} showing that this leads to a contradiction with ETH.

The reduction is similar to that from \cite{montanaro2017quantum} showing that counting the number of zeros of a degree-3 polynomial is \#\P-complete. Given a $k$-\texttt{SAT} instance $\phi$ with $n$ variables and $m$ clauses, we can use the sparsification lemma to assume that $m$ is $O(n)$ \cite{impagliazzo2001problems,dell2014exponential}. Then we introduce one additional variable $x_{n+1}$ and examine the formula $\phi' = x_{n+1}(1-\phi)$. Note that $\phi$ is satisfiable if and only if $\phi'$ is not balanced. There is a quantum circuit $C$ made up only of $O(m)$ $CCZ$ and $O(m)$ Hadamard gates, that computes the value of $\phi'(z)$ into an auxiliary register for any input $z$ on the first $n+1$ qubits. The circuit also requires $O(m)$ ancilla qubits that begin and end in the $\ket{0}$ state. As described in \cite{montanaro2017quantum}, the circuit can be associated with a degree-3 polynomial $f$ --- the $H$ gates are replaced by gadgets similar to that in Figure \ref{fig:qaoagadget}, introducing more ancilla qubits but turning the circuit into an IQP circuit --- that has $O(n+m)$ variables, such that $\gap(f) = \gap(\phi')$. Thus, given any constant $c$, we can choose $\ptNc $ small enough such that a $O(2^{\ptNc  n})$ time algorithm for determining whether $f$ is not balanced implies a $O(2^{cn})$ algorithm for $k$-\texttt{SAT}. Since we assumed the existence of the former, ETH is contradicted, proving the claim. 
\end{proof}

A variant of this claim shows that, like computing the permanent, computing the number of zeros to a degree-3 polynomial over $\mathbb{F}_2$ cannot be done in subexponential time, assuming \#ETH. This observation provides a link between poly3-NSETH($\ptNc $) and per-int-NSETH($\piNc $).

In comparison to poly3-NSETH($\ptNc $), per-int-NSETH($\piNc $) has advantages and disadvantages. There is no analogous black-box argument we can make for per-int-NSETH($\piNc $). On the other hand, there is no known non-trivial algorithm that rules out the conjecture for any $\piNc < 1$, making it possible that solving \texttt{per-int-NONZERO} with Ryser's formula is essentially optimal. The possible optimality of Ryser's formula is also bolstered by work in \cite{jerrum1982some}, where it is unconditionally proven that a monotone circuit requires $n(2^{n-1}-1)$ multiplications to compute the permanent, essentially matching the complexity of Ryser's formula.  This was recently extended to show similar lower bounds on monotone circuits that estimate output amplitudes of quantum circuits~\cite{monotone-sim-LB}.
Of course, per-int-NSETH($1-\delta$) for vanishing $\delta$ goes further and asserts that computation via Ryser's formula is optimal even with the power of non-determinism.  Thus our conjecture formalizes the statement that non-determinism cannot significantly speed up computing whether the permanent is nonzero.

\section{Simulation algorithms with additive error}\label{sec:adderror}

\subsection{Overview} 

The conclusions in the previous section state that any classical algorithm running in time less than the stated lower bound cannot sample from a distribution that has constant multiplicative error with respect to the true output distribution of the quantum circuit. However, real-world near-term quantum devices also cannot sample from such a distribution since each gate they perform has imperfect fidelity. It is more accurate to model the device distribution as having some constant \textit{additive} error with respect to the true circuit distribution, as in the sense of Eq.~\eqref{eq:additiveerror}. Thus we would like to also rule out classical simulation algorithms with constant additive error. At first glance, this seems challenging since the quantum advantage our results exploit rests on the ability to construct a quantum state with an acceptance probability that is \textit{exactly} 0 in one case and not 0 in another. Constant multiplicative error preserves whether or not an output probability is 0, but constant additive error does not. 

With minor modifications, our multiplicative-error analysis could be made robust to a situation where the additive error on individual output probabilities is exponentially small. When the total additive error is a constant (i.e.~Eq.~\eqref{eq:additiveerror} with $\epsilon = O(1)$), the typical amount of error on a randomly chosen output probability will indeed be exponentially small (since there are exponentially many possible outputs); however, it is possible that a constant amount of error could be concentrated on the particular output in the distribution that dictates when to accept, while most others have exponentially small error. This issue, or something similar to it, arises whether one uses $\NQP = \co\C_=\P$ or $\mathsf{PostBQP}=\PP$ in their argument for quantum computational supremacy.

%The constant additive error, as in Eq.~\eqref{eq:additiveerror}, is spread over exponentially many possible outputs, so the typical amount of error for an output probability is exponentially small. It would be possible to make a more robust version of our argument that could handle this situation if it could be guaranteed that the amount of error on the particular output we care about is typical. However, constant total error leaves open the possibility that a constant amount of error could be concentrated on a certain output (while most other outputs have exponentially small error). This issue, or something similar to it, arises whether one uses $\NQP = \co\C_=\P$ or $\mathsf{PostBQP}=\PP$ in their argument for quantum computational supremacy.

In previous work on QCS, this situation has been
handled by hiding the accepting output randomly among the exponentially many possible outputs, which makes it highly likely that the amount of error on that output probability is typical. Then, it must be conjectured that the hard classical problem solved by the quantum circuit is still hard in the average case; i.e.~it cannot be solved quickly even when the algorithm is allowed to fail on some small fraction of the instances. More concretely, previous work conjectured that it is $\#\P$-hard to approximate $\gap(f)^2$ on some sufficiently large constant fraction of $n$-variable degree-3 polynomial instances $f$ \cite{bremner2016average}, or analogously, that it is $\#\P$-hard to approximate $\Per(A)^2$ on some sufficiently large constant fraction of $n \times n$-matrix instances $A$ \cite{aaronson2011computational}.

We present a fine-grained version of this argument. As in the multiplicative case, we avoid a direct fine-graining of the exact logical path relying on $\PP = \PostBQP$ taken by \cite{aaronson2011computational,bremner2016average}. The central reason for this choice is so that we can avoid the step in the analysis where one classically estimates the acceptance probability of a classical randomized algorithm using resources that fall within the polynomial hierarchy (Stockmeyer's theorem \cite{stockmeyer1983complexity}). In particular, this estimation can be done efficiently with an oracle in the second level of the $\PH$, i.e.~$\P^{\Sigma_2^\P}$, or with randomness and an $\NP$ oracle, i.e.~$\BPP^\NP$. This is a costly step in the fine-grained world since constructing this algorithm explicitly would introduce significant overhead that would make our fine-grained lower bounds, and hence qubit estimates, worse. Additionally, the fine-grained conjecture we would need to make would pertain to algorithms lying in the third level of the exponential-time hierarchy (the natural generalization of the $\PH$) where little is known about fine-grained algorithms.

However, it is apparent that using the $\NQP=\co\C_=\P$ route will also be insufficient for this purpose. One reason was previously mentioned: $\NQP$ is not robust to additive error. But another important reason is that the problem $\texttt{poly3-NONBALANCED}$ is an average-case easy problem simply because all but an exponentially small fraction of degree-3 polynomials are not balanced (under plausible assumptions about the distribution of $\gap(f)$ with respect to random $f$). Hence one can construct a deterministic classical algorithm that outputs YES for every input instance $f$ and succeeds on an overwhelmingly large fraction of instances. 

Instead, we will use a third relationship between quantum and classical counting complexity to underlie our fine-grained, additive-error analysis: $\mathsf{SBQP} = \A_0\PP$ \cite{kuperberg2015hard}. The class $\mathsf{SBQP}$ is the quantum analogue of $\SBP$, which is, in a sense, an error-robust version of $\NP$ that formally lies between $\NP$ and $\MA$ \cite{bohler2006error}, within the second level of the $\PH$. Meanwhile $\A_0\PP$ is a classical counting complexity class that is hard for the $\PH$ \cite{vyalyi2003qma,kuperberg2015hard}, similar to $\PP$ and $\co\C_=\P$. In the context of QCS arguments, the class $\mathsf{SBP}$ was first used in \cite{fujii2014impossibility,fujii2015power}.

The specific counting problem we will solve with quantum circuits we call \texttt{poly3-SGAP} and still pertains to the gap of degree-3 polynomials over $\mathbb{F}_2$, but the question is different: now, given $f$, we ask if $\gap(f)^2$ is less than $2^{n-2}$ or if it is at least $2^{n-1}$, promised that one or the other is the case. We are able to prove that this problem is not a naively easy average-case problem in the way $\texttt{poly3-NONBALANCED}$ is, and we show how the ability to produce samples from quantum circuits with some small constant additive error would lead to a classical algorithm with ``$\mathsf{SB}$ power'' (replacing non-deterministic power in the analysis from previous sections) that solves the problem on a large fraction of instances. Our average-case conjecture asserts that any classical algorithm with $\mathsf{SB}$ power that succeeds on a large fraction of instances must have at least some exponential runtime, and this yields a lower bound on the sampling time for IQP and QAOA circuits. Beyond providing us with a fine-grained lower bound from which we can estimate the number of qubits for QCS, this analysis reveals a new perspective on additive-error QCS arguments --- one that avoids the need for Stockmeyer counting by replacing estimation of the quantity $\gap(f)^2$ with a decision problem of similar flavor.

We do not perform a fine-grained analysis of additive-error simulations for boson sampling circuits. However, we believe it would be possible to perform such an analysis with some additional work. This would require conjecturing the fine-grained average-case hardness of deciding whether $\Per(A)^2$ is at least a certain threshold, or at most a smaller threshold, for random matrices $A$. While there are no fundamental barriers to this analysis, the details would be more complicated than for IQP and QAOA, stemming from two factors. First, for IQP we can draw from a discrete set, the set of degree-3 polynomials. However, for boson sampling in the additive-error setting, we must draw from a continuous set, the set of Gaussian complex matrices. In a fine-grained analysis, we would need to describe a specific implementation that discretizes this set. Second, and more importantly, unlike the hiding mechanism for IQP, the hiding mechanism for boson sampling is only approximate, and it is conceptually more complicated than for IQP. For boson sampling, given a matrix $A$, one hides which amplitude they are interested in by drawing a large unitary $U$ randomly from the Haar distribution, and post-selecting on choices for which the matrix $A$ is a submatrix of $U$. If $U$ is sufficiently larger than $A$, and the entries of $A$ are i.i.d. Gaussian, then the position of the submatrix $A$ in $U$ will be approximately hidden \cite{aaronson2011computational}.

\subsection{Quantum computational supremacy with the complexity class \texorpdfstring{$\mathsf{SBP}$}{SBP}}
We now provide the definition of the classical complexity class $\SBP$. A language $L$ is in $\SBP$ if there exist polynomials $p$ and $q$, a constant $c>1$, and a polynomial-time Turing machine $M$ such that for all $x \in \{0,1\}^*$,
\begin{eqnarray}
x \in L  &\implies & |\{ y \in \{0,1\}^{p(|x|)}: M(x,y) = 1\}| \geq 2^{q(|x|)}\nn
x \not\in L &\implies& |\{ y \in \{0,1\}^{p(|x|)}: M(x,y) = 1\}|  \leq \frac{2^{q(|x|)}}{c}.\nonumber
\end{eqnarray}
Thinking of $M$ as a randomized machine with random bits $y$, we see that the acceptance probability of $M$ must be greater than some exponentially small threshold $2^{-(p(|x|)-q(|x|))}$ when the input $x$ is in $L$, and the acceptance probability must be smaller than some threshold that is a factor of $c$ smaller on all inputs $x$ not in the language. One may view this as a more error-robust version of $\NP$, which can also be worded in terms of thresholds: the machine must accept with at least some exponentially small probability (corresponding to a single computation path) when $x \in L$ and at most some smaller probability (namely 0) when $x \not\in L$. This interpretation makes it clear that the class $\SBP$ contains $\NP$. It is also known that $\SBP$ is contained in the second level of the $\PH$ \cite{bohler2006error}. Just as we might say a classical algorithm has the power of non-determinism to mean that it accepts its input if at least one of its computation paths accepts, we will say a classical algorithm has ``$\mathsf{SB}$ power'' and mean that it accepts its input as long as some exponentially small fraction of its computation paths accept, and rejects if fewer than $1/c$ times as many computation paths accept. Of course, there is a third possibility, that the number of accepting paths lies in the window between the two thresholds, and in this case the algorithm with $\mathsf{SB}$ power can act arbitrarily.

The quantum complexity class $\mathsf{SBQP}$ is the quantum generalization, defined roughly as languages that can be solved by polynomial-time quantum computation with $\mathsf{SB}$ power. It was shown that $\mathsf{SBQP} = \A_0\PP$ \cite{kuperberg2015hard} where $\A_0\PP$ is a classical counting complexity class defined \cite{vyalyi2003qma} to contain languages $L$ for which there are polynomials $p$ and $q$, a constant $c>1$, and a polynomial-time Turing machine $M$ such that 
\begin{eqnarray}
x \in L  &\implies & \gap(M(x, \cdot)) \geq 2^{q(|x|)}\nn
x \not\in L &\implies& \gap(M(x, \cdot)) \leq \frac{2^{q(|x|)}}{c}.\nonumber
\end{eqnarray}
The similarity to the definition of $\SBP$ is apparent. 

Formally speaking, we will actually be using promise versions of these complexity classes. For promise classes, languages are replaced by pairs of disjoint sets $(L_{YES}, L_{NO})$, where acceptance criteria must apply for inputs in $L_{YES}$ and rejection criteria for inputs in $L_{NO}$ but there may be inputs that lie in neither set (and there are no requirements on how the Turing machine or quantum circuits act on these inputs).

It was shown that the class $\mathsf{PromiseSBQP}=\mathsf{PromiseA}_0\PP$ is hard for the $\PH$, i.e.~$\PH \subseteq \P^\mathsf{PromiseSBQP}$ \cite{kuperberg2015hard}. This yields an alternate route to quantum computational supremacy: if one could classically simulate quantum circuits (with multiplicative error), then $\mathsf{PromiseSBQP}=\mathsf{PromiseSBP}$, which would imply the $\PH$ collapses to the third level. By exploiting $\mathsf{NQP} \subseteq \mathsf{SBQP}$ this statement could be improved to yield a collapse to the second level \cite{fujii2014impossibility, fujii2015power}.

\subsection{The problem \texttt{poly3-SGAP}}

To perform a fine-grained analysis we pick a specific problem suited to an analysis using $\SBP$, which we call \texttt{poly3-SGAP}. This problem plays the role of $\texttt{poly3-NONBALANCED}$ from the previous sections.

The input to $\texttt{poly3-SGAP}$ is a degree-3 polynomial $f$ over $\mathbb{F}_2$ with $n$ variables and no constant term, and the output should be YES if $(\gap(f)/2^n)^2 \geq 2^{-n-1}$, and NO if $(\gap(f)/2^n)^2 \leq 2^{-n-2}$. This is a promise problem; we are promised that the input $f$ does not satisfy $2^{-n-2} < (\gap(f)/2^n)^2 < 2^{-n-1}$. Refer to Figure \ref{fig:gapfdist} for a schematic of the division between YES and NO instances under the assumption that $\gap(f)$ is distributed as a Gaussian random variable. We let the set $\mathcal{F}_n^{\text{prom}}$ refer to all degree-3 polynomials with $n$ variables and no constant term that satisfy the promise, and if $f \in \mathcal{F}_n^{\text{prom}}$, we let $S(f) = 0$ if the answer to \texttt{poly3-SGAP} is NO and $S(f)=1$ if it is YES; thus \texttt{poly3-SGAP} is simply the task of computing the (partial) function $S$.

\begin{figure}[ht]
    \centering
    \includegraphics[width=0.95\columnwidth]{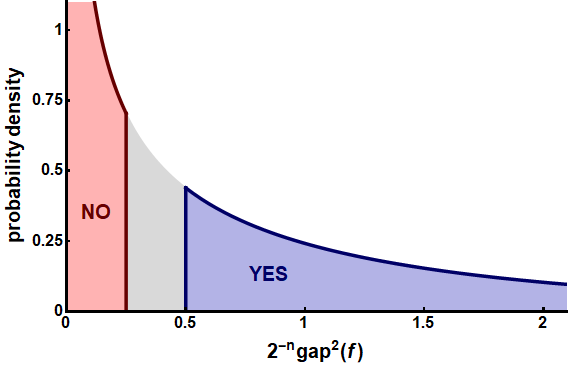}
    \caption{Distribution of $2^{-n}\gap(f)^2$ for randomly drawn $f$, under the assumption that $\gap(f)$ is distributed as a Gaussian with mean 0 and standard deviation $2^{n/2}$. For polynomials $f$ falling in the red (blue) area, occupying roughly 38\% (48\%) of the distribution, the answer to \texttt{poly3-SGAP} should be NO (YES). Polynomials falling in the gray area do not satisfy the promise. }
    \label{fig:gapfdist}
\end{figure}

Since one can reduce any Boolean function to a degree-3 polynomial over $\mathbb{F}_2$, \texttt{poly3-SGAP} is a prototypical $\mathsf{PromiseSBQP}$ problem. Indeed \texttt{poly3-SGAP} captures the hardness of $\mathsf{PromiseSBQP}$ since a \texttt{poly3-SGAP} oracle is sufficient to simulate a $\#\P$ oracle, which can be seen by adapting the proof in \cite{kuperberg2015hard} that $\P^{\#\P} = \P^{\mathsf{PromiseSBQP}}$. 

By the correspondence between IQP circuits and degree-3 polynomials, it is also clear that \texttt{poly3-SGAP} can be solved by an IQP circuit equipped with $\mathsf{SB}$ power, using only $n$ qubits and $g_1(n)$ gates. Given the input $f$, one simply needs to run the circuit $C_f$, whose acceptance probability is exactly $(\gap(f)/2^n)^2$; the acceptance probability is greater than $2^{-n-1}$ when the YES criteria are met, and it is less than $2^{-n-2}$ when the NO criteria are met.

Unlike non-determinism, $\mathsf{SB}$ power affords the brute-force algorithm a quadratic speedup; $\texttt{poly3-SGAP}$ can be solved in $\text{poly}(n)2^{0.5n}$ classical $\mathsf{SB}$ time. How this is done is discussed later in Theorem \ref{thm:oracle-ave-LB}.

\subsection{Lower bounds for additive-error simulation}

\subsubsection{For IQP circuits}

Deriving fine-grained lower bounds for the additive-error case will follow the same structure as the multiplicative-error case, with a few differences. We begin by assuming we have a classical simulation algorithm that for some fixed constant $\epsilon$, produces samples up to additive error $\epsilon$ from any IQP circuit  with $q$ qubits and $g_1(q)$ internal gates in $s_4(q)$ time steps. We will use this classical algorithm as a subroutine to solve \texttt{poly3-SGAP} in the average case.

Given as input a random instance $f$ from $\mathcal{F}_n^{\text{prom}}$, we let $\bar{f}$ refer to the degree-3 polynomial that has no degree-1 terms, but whose degree-2 and degree-3 terms are the same as $f$. Moreover let the set $[\bar{f}]$ contain all polynomials $g$ for which $\bar{f}=\bar{g}$. Let the $n$-bit string $\delta^f\in \{0,1\}^n$ be defined such that $\delta^f_i = 1$ if and only if the degree-1 term $x_i$ appears in the polynomial $f(x)$. Thus $f(x) = \bar{f}(x) + \delta^f \cdot x$, where $\delta^f \cdot x = \sum_i \delta^f_ix_i$. We construct the IQP circuit $C_{\bar{f}}$ which has $n$ qubits and $g_1(n)$ gates. It was noted previously that 
\begin{equation}
    \bra{\bar{0}}U_{\bar{f}} \ket{\bar{0}} = \gap(\bar{f})/2^n,
\end{equation}
where $U_f$ is the unitary implemented by the circuit $C_f$ and $\ket{\bar{0}}$ is the all $\ket{0}$ starting state. This can easily be generalized to additionally show that
\begin{equation}
    \bra{\delta^f}U_{\bar{f}} \ket{\bar{0}} = \gap(f)/2^n,
\end{equation}
that is, the probability of measuring outcome $\delta^f$ after running $U_{\bar{f}}$ is exactly proportional to $\gap(f)^2$.

Thus we construct the following classical algorithm $\mathcal{A}$: given $f$, we use our classical simulation algorithm to produce a sample from an output distribution that approximates that of $C_{\bar{f}}$ up to $\epsilon$ additive error. If outcome $\delta^f$ is obtained, $\mathcal{A}$ accepts, and otherwise $\mathcal{A}$ rejects. For any $z \in \{0,1\}^n$, let $P_{\bar{f}}(z)$ be the probability the classical simulation algorithm produces outcome $z$ on the circuit $C_{\bar{f}}$ and let $Q_{\bar{f}}(z)$ be probability associated with the true distribution. So the acceptance probability of the classical algorithm we constructed is exactly $P_{\bar{f}}(\delta^f)$. We claim that the algorithm $\mathcal{A}$, when considered as an algorithm with so-called $\mathsf{SB}$ power, solves \texttt{poly3-SGAP} on a large fraction ($1-60\epsilon$) of instances. This is formally stated as follows.

\begin{theorem}
For $f \in \mathcal{F}_n^{\text{prom}}$, let $\mathcal{A}(f)$ evaluate to 1 if $S(f) = 1$ and $P_{\bar{f}}(\delta^f) \geq 2^{-n-1}*5/6$ or if $S(f)=0$ and $P_{\bar{f}}(\delta^f) \leq 2^{-n-1}*2/3$. Otherwise it evaluates to 0. Then the fraction of $f \in \mathcal{F}_n^{\text{prom}}$ for which $\mathcal{A}(f) = 1$ is at least $1-60\epsilon$.
\end{theorem}
\begin{proof}
Extend the function $\mathcal{A}$ to the entire set of degree-3 polynomials with $n$ variables by letting $\mathcal{A}(g)=1$ if $g \not\in \mathcal{F}_n^{\text{prom}}$.
Now suppose the statement were false. By Lemma \ref{lem:promisefraction}, below, the fraction of degree-3 polynomials in $\mathcal{F}_n^{\text{prom}}$ that satisfy the promise is at least 1/5; hence, the fraction of all degree-3 polynomials for which $\mathcal{A}(f)=0$ is at least $12\epsilon$. Thus, there must be some $g$ where the fraction of $g' \in [\bar{g}]$ for which $\mathcal{A}(g')=0$ is at least $12 \epsilon$. For each such $g'$ with $\mathcal{A}(g')=0$, it is either the case that $S(g') = 1$, $Q_{\bar{g}}(\delta^{g'}) \geq 2^{-n-1}$, and $P_{\bar{g}}(\delta^{g'}) < 2^{-n-1}*5/6$ or that $S(g') = 0$, $Q_{\bar{g}}(\delta^{g'}) \leq 2^{-n-1}*1/2$, and $P_{\bar{g}}(\delta^{g'}) > 2^{-n-1}*2/3$. In particular, $\lvert Q_{\bar{g}}(\delta^{g'})-P_{\bar{g}}(\delta^{g'}) \rvert \geq 2^{-n}/12$ for all such $g'$, and the number of $g'$ for which this is true is at least $2^n*12\epsilon$. Thus $\sum_{z}\lvert P_{\bar{g}}(z)-Q_{\bar{g}}(z) \rvert > \epsilon$, contradicting that the classical simulation algorithm has additive error at most $\epsilon$. 
\end{proof}

\begin{lemma}\label{lem:promisefraction}
For all $n$, the fraction of degree-3 polynomials $f$ with $n$ variables for which $f \in \mathcal{F}_n^{\text{prom}}$ is at least 1/5.
\end{lemma}
\begin{proof}
In Ref.~\cite{bremner2016average}, it was shown that $\mathbb{E}((\gap(f)/2^n)^4) \leq 3 * 2^{-2n}$, where the expectation is taken over all degree-3 polynomials with $n$ variables. Moreover, $\mathbb{E}((\gap(f)/2^n)^2) = 2^{-n}$. The Paley-Zygmund inequality states that for a random variable $Z$, $\Pr(Z > \theta \mathbb{E}(Z)) \geq (1-\theta)^2 \mathbb{E}(Z)^2/\mathbb{E}(Z^2)$. Taking $Z = \gap(f)^2/2^{2n}$ and $\theta = 1/2$, we find that $\Pr((\gap(f)/2^n)^2 \geq 2^{-n-1}) \geq 1/12$. Additionally, in Lemma \ref{lem:massbound} in Appendix \ref{app:moments}, we extend the method of \cite{bremner2016average} to compute more moments and prove that $\Pr((\gap(f)/2^n)^2 \leq 2^{-n-2}) \geq 0.12$. Taken together, this shows that at least 1/5 of the instances satisfy the promise. 
\end{proof}

Now we make an average-case fine-grained conjecture about $\mathsf{SB}$ algorithms that solve \texttt{poly3-SGAP}.

\begin{conj}[poly3-ave-SBSETH($\ptSBc$)]\label{conj:poly3aveSBSETH}
    Suppose a classical probabilistic algorithm (in the Word RAM model of computation) runs in $p(n)$ time steps and has the following property: there exists a function $q(n) \leq p(n)$ and a constant $c$ such that for at least a 11/12 fraction of instances $f$ from $\mathcal{F}_n^{\text{prom}}$, either $\gap(f)^2 \geq 2^{n-1}$ and the probability the algorithm accepts $f$ is at least $2^{-q(n)}$ or $\gap(f)^2 \leq 2^{n-2}$ and the probability the algorithm accepts $f$ is at most $2^{-q(n)}/c$. Then, the algorithm must satisfy $p(n) \geq 2^{\ptSBc n-1}$. That is, there is no classical algorithm with $\mathsf{SB}$ power that solves $\texttt{poly3-SGAP}$ and runs in fewer than $2^{\ptSBc n-1}$ timesteps.
\end{conj}

As long as $\epsilon < 1/720$, the algorithm $\mathcal{A}$ is exactly such an algorithm, with runtime $s_4(n)$. Thus, assuming poly3-ave-SBSETH($\ptSBc$), we obtain the lower bound
\begin{equation}\label{eq:iqpaddbound}
    s_4(n) \geq 2^{\ptSBc n-1}.
\end{equation}

The choice $11/12 \approx 0.92$ in Conjecture \ref{conj:poly3aveSBSETH} is motivated by a rigorous bound on the maximum success probability of a naive algorithm that always outputs the same answer on the \texttt{poly3-SGAP} problem. We discuss this further in Section \ref{sec:additiveevidence} and prove this bound in Appendix \ref{app:moments}. If we assume that the gap of random degree-3 polynomials is distributed as a Gaussian, as depicted in Figure \ref{fig:gapfdist}, this bound would be improved to roughly $0.56$. Moreover, the same assumption would allow the result in Lemma \ref{lem:promisefraction} to be improved from 1/5 to $0.86$. Assuming the Gaussian gap distribution and replacing 11/12 with 0.56 in Conjecture \ref{conj:poly3aveSBSETH}, the error tolerance of our additive-error approximation algorithm could be improved from $\epsilon<1/720$ to $\epsilon < 1/32$, which would be significantly more attainable experimentally.

\subsubsection{For QAOA circuits}

We previously demonstrated how one turns the IQP circuit $C_f$ into a QAOA circuit by introducing $n$ ancilla qubits. Using this construction as the only modification to the previous analysis yields the following lower bound. 

\begin{equation}\label{eq:qaoaaddbound}
    s_5(n) \geq 2^{\frac{\ptSBc n}{2} - 1},
\end{equation}
where $s_5(n)$ is the number of time steps for a hypothetical algorithm that simulates any QAOA circuit with $n$ qubits and $g_2(n)$ gates, up to additive error at most $\epsilon < 1/720$.

\subsection{Evidence for average-case conjecture}\label{sec:additiveevidence}

Our fine-grained lower bounds derived in the previous subsection rest on poly3-ave-SBSETH($\ptSBc$), a conjecture that asserts the nonexistence of fine-grained exponential-time classical algorithms with $\mathsf{SB}$ power that solve a hard problem in the average case. There are two main differences between poly3-ave-SBSETH($\ptSBc$) and the previously discussed poly3-NSETH($\ptNc$). First, poly3-ave-SBSETH($\ptSBc$) pertains to $\mathsf{SB}$ classical algorithms instead of non-deterministic classical algorithms. The power of $\mathsf{SB}$ is at least as powerful as non-determinism: $\SBP$ contains $\NP$. Yet, $\SBP$ still lies only in the second level of the $\PH$, somewhere between $\NP$ and $\MA$. This aspect makes poly3-ave-SBSETH($\ptSBc$) less plausible than poly3-NSETH($\ptNc$) when $\ptSBc = \ptNc$; as we will see later, the move from non-determinism to $\mathsf{SB}$ yields a quadratic speedup of the black-box query complexity, but the similarity between the two models makes additional speedups based on this aspect unlikely. The second and more substantial difference is that poly3-SBSETH($\ptSBc$) is about an average-case problem, while poly3-NSETH($\ptNc$) is about a worst-case problem. Indeed, while the problem \texttt{poly3-SGAP} is known to be hard for the $\PH$ in the worst case, its complexity is unknown in the average case, much less its fine-grained complexity.

Note that in \cite{ball2017average} and \cite{goldreich2017worst}, average-case fine-grained complexity was studied in the context of polynomial-time algorithms; they defined various average-case problems in $\P$ and showed that standard worst-case fine-grained conjectures like SETH imply lower bounds for these problems that essentially match known upper bounds. Ideally, we would be able to show something similar for our problem, but we are unable to do so due to the fact that our problem pertains to $\mathsf{SB}$ algorithms for $\#\P$-hard problems. 

Nonetheless, in the remainder of this section, we strive to put the average-case aspect of poly3-ave-SBSETH($\ptSBc$) in context and provide as much evidence as possible in its favor.

A necessary condition for \texttt{poly3-SGAP} to be average-case hard is that naive determinstic algorithms that produce the same output on every input cannot solve it on more than a constant fraction of inputs. The maximum success probability of such an algorithm is the fraction of instances for which the output should be YES, or the fraction for which the output should be NO, whichever is larger. In the limit as $n$ becomes large, this is given mathematically by
\begin{equation}\label{eq:naiveprob}
p_0 = \liminf_{n\rightarrow \infty}\left(\max_{j=0,1}\left( \Pr_{f \leftarrow \mathcal{F}_n^{\text{prom}}}(S(f)=j)\right)\right).
\end{equation}
We would say the problem is naively average-case easy if $p_0 = 1$. If, for large $n$, $\gap(f)/2^n$ were distributed as a Gaussian with mean 0 and variance $2^{-n}$ as illustrated in Figure \ref{fig:gapfdist}, then $p_0 = 0.48/(0.38+0.48) = 0.56$. In Appendix \ref{app:moments}, we show that, in the limit of large $n$, the moments $\mathbb{E}((\gap(f)/2^n)^{2k})$ approach $\mathbb{E}(x^{2k})$ when $x$ is distributed as a Gaussian with mean 0 and variance $2^{-n}$. Although we do not formally prove that this implies the distribution itself approaches a Gaussian distribution, we take it as evidence in favor of that hypothesis. Moreover, as we also show in Appendix \ref{app:moments}, our calculation of the moments formally implies that $p_0 \leq 11/12$, indicating that $\texttt{poly3-SGAP}$ is not naively average-case easy. This contrasts with the problem \texttt{poly3-NONBALANCED} for which we expect it to be the case that $p_0 = 1$; the fraction of instances for which $\gap(f)=0$ vanishes with increasing $n$.

It remains possible a more sophisticated algorithm could solve \texttt{poly3-SGAP} efficiently. We now provide arguments showing how certain approaches to find such an algorithm are doomed for failure.

We begin by extending the black-box argument from Section \ref{sec:evidence} so that it rules out $\mathsf{SB}$ algorithms solving the $\texttt{SGAP}$ problem in the average case for random general Boolean functions on $n$ variables:

\begin{customthm}{\ref{thm:oracle-ave-LB} (informal)}
A black-box algorithm equipped with $\mathsf{SB}$ power that solves the \texttt{SGAP} problem on $0.56+\epsilon$ fraction of instances for some $\epsilon>0$ must make at least $\Omega(2^{n/2})$ queries, and moreover there is a matching upper bound.
\end{customthm}

\begin{theorem}\label{thm:oracle-ave-LB}
Let $\mathcal{H}_n$ contain all $n$-bit Boolean functions for which either $(\gap(h)/2^n)^2 \geq 2^{-n-1}$ or $(\gap(h)/2^n)^2 \leq 2^{-n-2}$. Suppose a randomized algorithm $\mathcal{A}$ with black-box access to $h \in \mathcal{H}_n$ has the following properties: (1) $\mathcal{A}$ makes $L$ queries to $h$ and (2) there exists a function $t(n)$ and constants $n_0$, $c>1$, and $\epsilon > 0$, such that whenever $n \geq n_0$ and for a fraction at least $0.56+\epsilon$ of $h \in \mathcal{H}_n$, $\mathcal{A}(h)$ accepts with at least probability $t(n)$ if $(\gap(h)/2^n)^2 \geq 2^{-n-1}$, and accepts with at most probability $t(n)/c$ if $(\gap(h)/2^n)^2 \leq 2^{-n-2}$. Then $L \geq \Omega(2^{n/2})$. Moreover, there exists such an algorithm with $L = O(2^{n/2})$.
\end{theorem}

\begin{proof}
The main idea behind the proof is similar to that of Theorem \ref{thm:oracle-LB}: we take YES instances on which the algorithm succeeds, minimally modify the instance so that it is a NO instance, and argue that the algorithm must behave the same way (unless it makes many queries). Thus it cannot succeed on a large fraction of both the YES and the NO instances. The implementation here is more complicated than in Theorem \ref{thm:oracle-LB} because the algorithm we are dealing with is a $\mathsf{SB}$ algorithm and additionally because we must worry about the fraction of instances on which the algorithm can succeed, and not just whether or not it succeeds on all instances.

The distribution of $\gap(h)/2^n$ for $h$ a uniformly random Boolean function is a shifted and scaled binomial distribution with mean $0$ and variance $2^{-n}$, since each of the $2^n$ entries in the truth table is 0 or 1 with equal probability. As $n \rightarrow \infty$, it approaches a Gaussian distribution with the same mean and variance. Let $N \subset \mathcal{H}_n$ contain instances $h$ for which $(\gap(h)/2^n)^2 \leq 2^{-n-2}$, and let $Y \subset \mathcal{H}_n$ contain instances $h$ for which $2^{-n-1} \leq (\gap(h)/2^n)^2 \leq \gamma2^{-n}$, where $\gamma \approx 2.76$ is the constant for which $\erf(\sqrt{2}/4) + \erf(1/2) = \erf(\sqrt{2\gamma}/2)$, defined as such so that $\lvert N \rvert/\lvert Y \rvert \rightarrow 1$ as $n\rightarrow \infty$.

There exists an algorithm that simply always outputs YES and succeeds on a fraction 0.56 (as $n \rightarrow \infty$) of instances in $\mathcal{H}_n$  --- the instances in $\mathcal{H}_n \setminus N$ ---  without making any queries. If $\mathcal{A}$ succeeds on a fraction $0.56 + \epsilon$ of instances in $\mathcal{H}_n$, then it must succeed on at least a fraction $0.5+\epsilon$ of instances in $N \cup Y$, and we will show that this implies it makes at least $\Omega(2^{n/2})$ queries.

Fix an arbitrary choice of $t(n)$, $c$, and $\epsilon$. Consider an instance $h \in Y$ and an even integer $x$ such that $(x/2^n)^2 \leq 2^{-n-2}$. Let $N_{h,x}$ include all instances $g \in N$ for which $\gap(g)/2^n = x$ and $g$ and $h$ differ on exactly $\lvert \gap(g)-\gap(h) \rvert /2 \leq O(2^{n/2})$ of the $2^n$ input strings. That is, $g$ can be formed from $h$ solely by switching entries in the truth table from 0 to 1 when $\gap(h)>0$ or from 1 to 0 when $\gap(h)<0$. Consider the following process: first draw an instance $h$ uniformly at random from $Y$ and choose an integer $x$ according to the binomial distribution with mean 0 and variance $2^n$ (rejecting and repeating until $(x/2^n)^2 \leq 2^{-n-2}$), then with 1/2 probability output $h$ and with 1/2 probability output a uniformly random instance from $N_{h,x}$. The output of this procedure is an element from $N \cup Y$ and as $n\rightarrow \infty$, the distribution of outputs becomes arbitrarily close in $\ell_1$ distance to the uniform distribution.

Now suppose the algorithm $\mathcal{A}$ succeeds on input $h \in Y$ i.e.~its acceptance probability is at least $t(n)$. Viewing $\mathcal{A}$ as a deterministic machine taking as input $h$ as well as a uniformly random string $r \in R$, we have \begin{equation}
    |\{r \in R: \mathcal{A}(h,r) = 1\} |\geq t(n) |R|.
\end{equation}
Consider a fixed string $r$ for which $\mathcal{A}(h,r)=1$. The algorithm queries a fixed subset of size $L$ among the $2^n$ input strings. If (for any fixed choice of $x$) we choose a function $g \in N_{h,x}$ uniformly at random, then the set of input strings $z$ for which $g(z) \neq h(z)$ has cardinality $O(2^{n/2})$ and is equally likely to be any subset of that size containing only strings for which $h(z)=0$ (resp.~$h(z)=1$) when $\gap(h)>0$ (resp.~$\gap(h) < 0$). Thus, by the union bound, the probability over the choice of $g \in N_{h,x}$ that $\mathcal{A}(h,r)$ queries some string $z$ for which $h(z) \neq g(z)$ is at most $O(L 2^{-n/2})$. If all $L$ queried strings $z$ have $g(z) = h(z)$, then $\mathcal{A}$ proceeds the same on inputs $(g,r)$ as it does on inputs $(h,r)$ and outputs 1. This will be true for every $r$ for which $\mathcal{A}(h,r)=1$ and for every $x$. Thus we can say
\begin{eqnarray}
    &&\mathbb{E}_{g \in N_{h,x}}\left(|\{r\in R: \mathcal{A}(h,r) = \mathcal{A}(g,r) = 1\}|\right) \nonumber \\
    &\geq& \left(1-O(L 2^{-n/2})\right)|\{r\in R: \mathcal{A}(h,r) = 1\}|, \label{eq:expectationbound}
\end{eqnarray}
which implies
\begin{eqnarray}
    &&\Pr_{g \in N_{h,x}}\left(\frac{|\{r\in R: \mathcal{A}(h,r) = \mathcal{A}(g,r) = 1\}|}{|R|} \leq \frac{t(n)}{c}\right) \nonumber \\
    &\leq& \Pr_{g \in N_{h,x}}\left(\frac{|\{r\in R: \mathcal{A}(h,r) = \mathcal{A}(g,r) = 1\}|}{|\{r\in R: \mathcal{A}(h,r) = 1\}|} \leq \frac{1}{c}\right) \nonumber \\
    &\leq& O\left(\frac{L 2^{-n/2}}{1-1/c}\right),
\end{eqnarray}
where the last line follows from a rearrangement of Eq.~\eqref{eq:expectationbound} and Markov's inequality.
If $L = o(2^{n/2})$, then this quantity approaches 0 as $n \rightarrow \infty$. Now consider again the procedure of choosing a random instance $h$ from $Y$, a random $x$, and a random $g$ from $N_{h,x}$, then with equal chance running $\mathcal{A}$ on $h$ or running $\mathcal{A}$ on $g$. The bound above shows that this procedure can succeed with probability at most $1/2 + o(1)$. For any $\epsilon$, we may choose $n$ large enough so that the success fraction is smaller than $1/2+\epsilon/2$, and so that the procedure draws from a distribution that is $\epsilon/2$-close to uniform over $N \cup Y$, completing the proof that there is no algorithm satisfying the conditions in the statement with $L = o(2^{n/2})$. 

To show the upper bound, that $L = O(2^{n/2})$ is possible, we give a simple algorithm $\mathcal{A}$ which succeeds for every instance in $\mathcal{H}_n$ for any $n$. On an input $h$, the algorithm randomly selects $L=10*2^{n/2}$ input strings $\{z_i\}$ and queries $h(z_i)$ for each $z_i$. If $h(z_i) = h(z_j)$ for every $i,j$, then the algorithm accepts. Otherwise, it rejects. If $(\gap(h)/2^n)^2 \geq 2^{-n-1}$, then the probability of acceptance is at least
\begin{eqnarray}
\left(\frac{1}{2}+\frac{|\gap(h)|}{2^{n+1}}\right)^L &\geq& \frac{1}{2^L} \left(1 + 2^{(-n-1)/2}\right)^L \nonumber \\
&\geq& \frac{1}{2^L}\exp(L 2^{(-n-1)/2}*0.9) \nonumber \\
&=& 2^{-10*2^{n/2}}\exp(9/\sqrt{2}),
\end{eqnarray}
where in the second-to-last line we have used the bound $\log(1+u) \geq 0.9u$ when $u$ is sufficiently small.

Meanwhile, if $(\gap(h)/2^n)^2 \leq 2^{-n-2}$, then the probability of acceptance is at most
\begin{eqnarray}
2\left(\frac{1}{2}+\frac{|\gap(h)|}{2^{n+1}}\right)^L &\leq& \frac{2}{2^L} \left(1 + 2^{-n/2-1}\right)^L \nonumber \\
&\leq& 2*2^{-10*2^{n/2}} \exp(L 2^{-n/2-1}) \nonumber \\
&=& 2^{-10*2^{n/2}}*2\exp(5).
\end{eqnarray}
Thus, we may choose $t(n) = 2^{-10*2^{n/2}}\exp(9/\sqrt{2})$ and $c=1.5$ and then for every $h$, if $(\gap(h)/2^n)^2$ is above the threshold the acceptance probability is above $t(n)$ and if it is below the threshold the acceptance probability is below $t(n)/c$. This completes the proof. 
\end{proof}

The fact that the bound is tight formally rules out poly3-ave-SBSETH($\ptSBc$) for values of $\ptSBc > 0.5$. However, it indicates that refuting poly3-ave-SBSETH($\ptSBc$) for values of $\ptSBc < 0.5$ would require an algorithm that goes beyond simply evaluating the function $f(z)$ for various strings $z$; it would have to utilize the fact that $f$ is a random degree-3 polynomial over $\mathbb{F}_2$, and not a general Boolean function.

Note that the average-case aspect of the problem has no effect on the black-box query complexity in this case. The complexity is quadratically weaker ($\Theta(2^{n/2})$) than what we found for the \texttt{NONBALANCED} problem ($\Theta(2^n)$) in the worst case in Theorem \ref{thm:oracle-LB}, but this difference is a result of moving from non-deterministic algorithms to $\mathsf{SB}$ algorithms. It is \textit{not} a result of moving from the worst case to the average case. Indeed, the algorithm we give to show the upper bound $O(2^{n/2})$ on the query complexity in Theorem \ref{thm:oracle-ave-LB} solves the \texttt{SGAP} problem not only in the average case, but also in the worst case, demonstrating that there is no asymptotic difference between the average-case and worst-case black-box query complexity for $\mathsf{SB}$ algorithms solving the \texttt{SGAP} problem. We take this as evidence that the average-case aspect of poly3-ave-SBSETH($\ptSBc$) would not play a crucial role in its veracity.

We bolster this intuition by giving a worst-case-to-quasi-average-case reduction for counting zeros of degree-3 polynomials over $\mathbb{F}_2$. The ``quasi''-average-case problem we reduce to is to, for every degree-3 polynomial $f$, with high probability count the number of zeros to $g$ where $g$ is a random degree-3 polynomial whose degree-2 and degree-3 parts agree with $f$. This problem is harder than the fully average-case version where the degree-2 and degree-3 parts are also randomized. 

\begin{theorem}\label{thm:worstcaseavecasereduction}
    For any degree-3 polynomial $f$ with $n$ variables and no constant term, let $[\bar{f}]$ be the set of all degree-3 polynomials whose degree-2 and degree-3 parts agree with $f$. Suppose we have an algorithm $\mathcal{A}$ with runtime $T(n)$ that, for every $f$, computes $\gap(g)$ correctly on at least a fraction $1-1/(3n)$ of instances $g \in [\bar{f}]$. Then there is a randomized algorithm $\mathcal{A'}$ with runtime $nT(n)+\text{poly}(n)$ that computes $\gap(f)$ correctly for every $f$. 
\end{theorem}
\begin{corollary}
It is $\#P$-hard to, for every $f$, compute $\gap(g)$ on a fraction $1-1/(3n)$ of instances $g \in [\bar{f}]$.
\end{corollary}

\begin{proof}[Proof of Theorem \ref{thm:worstcaseavecasereduction}]
    Let $\mathcal{P}$ be the quasi-average-case problem described in the statement. We will show that an oracle to $\mathcal{P}$ can be used to efficiently compute $\gap(f)$ in the worst case. Given a degree-3 polynomial $f$ with $n$ variables, choose a polynomial $g \in [\bar{f}]$ uniformly at random and use the oracle to $\mathcal{P}$ to compute $\gap(g)$. Write the degree-1 polynomial $f(x)+g(x)$ as $\sum_j u_j x_j$ for some vector $u \in \mathbb{F}_2^n$. If $u$ is the 0 vector, then $f=g$ and we may simply output $\gap(g)$. Otherwise, there is an index $j$ for which $u_j = 1$. Without loss of generality assume $j=n$. Then we can perform the bijective linear change of variables $x_n' = \sum_j u_j x_j$ and $x_k' = x_k$ for all $k \neq n$. This yields new degree-3 polynomials $f'(x')$ and $g'(x')$ such that $\gap(f) = \gap(f')$, $\gap(g) = \gap(g')$ (since the transformation is bijective), and $g'(x') = f'(x') + x_n'$.  
    
    We can express
    \begin{eqnarray}
        \gap(g) &=& \gap(g') = \gap(f' | x_n' = 0) - \gap(f' | x_n' = 1)  \nonumber \\
        \gap(f) &=& \gap(f') = \gap(f' | x_n' = 0) + \gap(f' | x_n' = 1), \nonumber
    \end{eqnarray}
    where $f'|x_n'$ denotes the degree-3 polynomial on $n-1$ variables formed by taking $f'$ and fixing the value of $x_n'$. Thus,
    \begin{equation}
        \gap(f) = \gap(g) + 2 \gap(f'|x_n'=1).
    \end{equation}
    
    We recursively compute $\gap(f'|x_n'=1)$, unless it has just one variable, in which case we compute it by brute-force, and we add twice the result to $\gap(g)$ to produce our output $\gap(f)$. 
    
    Each level of recursion requires $\text{poly}(n)$ steps to perform the linear transformation as well as one oracle call, which takes at most $T(n)$ time steps. There are $n$ recursion levels, meaning the total runtime of the algorithm is only $nT(n) + \text{poly}(n)$.
    
    The algorithm will succeed as long as it computes $\gap(g)$ correctly at each of the $n$ levels of recursion. Since each of these occurs with probability at least $1-1/(3n)$, by the union bound the procedure as whole succeeds with at least $2/3$ probability. Note that this success probability is over the randomness of the algorithm itself --- it could be boosted by repetition --- and not over the randomness in the instance. 
\end{proof}

This worst-case-to-quasi-average-case reduction is subideal for several reasons. First it is only quasi-average case and not fully average-case. However, randomizing the instances only over $[\bar{f}]$ and not over the entire set of degree-3 polynomials actually fits fairly well with the framework of our analysis, since for a fixed IQP circuit $\mathcal{C}_f$ the probability of sampling each of the $2^n$ outputs corresponds with $(\gap(g)/2^n)^2$ for a different $g \in [\bar{f}]$. A larger issue is that the worst-case-to-quasi-average-case reduction only works for computing $\gap(f)$ and not for computing $\gap(f)^2$: if one knows $\gap(g)^2$ for a random $g \in [\bar{f}]$, it is not clear how to determine whether $\gap(g) = \pm \sqrt{\gap(g)^2}$, leading to an exponentially large tree of possible values for $\gap(f)$ after $n$ levels of recursion. Finally, like worst-to-average-case reductions for computing the permanent \cite{lipton1989new,aaronson2011computational} or the output probability of a random quantum circuit \cite{bouland2018quantum,movassagh2018efficient,movassagh2019cayley}, our reduction requires the computation of $\gap(g)$ at each step to be exact, as any errors would be uncontrollably amplified by the recursion process. However, we note that the idea behind our reduction of exploiting gap-preserving linear transformations of degree-3 polynomials over $\mathbb{F}_2$ is quite distinct from the strategy of polynomial interpolation that underlies these previous worst-case-to-average-case reductions. 

Despite these shortcomings, the reduction rules out certain approaches to refuting our conjecture poly3-ave-SBSETH($\ptSBc$). In particular, an algorithm that refutes it by exactly computing $\gap(f)$ --- like the algorithm from LPTWY \cite{lokshtanov2017beating} --- would need to succeed in both the average case and the worst case, or otherwise succeed in the average case but \textit{not} the quasi-average case.

Taken together, the black-box bound and the worst-case-to-quasi-average-case reduction bolster the plausibility of poly3-ave-SBSETH($\ptSBc$), at least relative to poly3-NSETH($\ptNc$), because they present roadblocks on ways to utilize the average-case nature of the former conjecture to refute it without also refuting the latter.

For our qubit calculations, we take $\ptSBc = 0.5$ since any larger number is formally refuted by the black-box result and any smaller number would yield a non-trivial classical algorithm.

\section{Number of qubits to achieve quantum computational supremacy}\label{sec:numqubits}

\subsection{Our estimate}
We can use the lower bounds on the runtime of a hypothetical classical simulation algorithm for IQP, QAOA, and boson sampling circuits in Eqs.~\eqref{eq:iqpbound},  \eqref{eq:qaoabound}, \eqref{eq:bsbound},  \eqref{eq:iqpaddbound}, and \eqref{eq:qaoaaddbound} to estimate the minimum number of qubits required for classical simulation of these circuit models to be intractable.

The fastest supercomputers today can perform \hl{at between $10^{17}$ and $10^{18}$} FLOPS (floating-point operations per second)\footnote{A list of the fastest supercomputers is maintained at  \url{https://www.top500.org/statistics/list/}.}. Using our lower bounds, we can determine the number of qubits/photons $q$ such that the lower bound on $s_i(q)$ is equal to $\hl{10^{18}}\cdot 60\cdot 60 \cdot 24 \cdot 365\cdot 100$, the maximum number of floating-point operations today's supercomputers can perform in one century, for $i=1,2,3\hl{,4,5}$. Using our multiplicative-error lower bounds, we calculate that for IQP circuits it is \hl{93}/$\ptNc $ qubits (from Eq.~\eqref{eq:iqpbound}), for \textsc{QAOA} circuits it is \hl{185}/$\ptNc $ qubits (from Eq.~\eqref{eq:qaoabound}), and for boson sampling circuits it is \hl{93}/$\piNc $ photons (from Eq.~\eqref{eq:bsbound}). \hl{The additive-error bounds replace $\ptNc$ by $\ptSBc$ (from Eqs.~\eqref{eq:iqpaddbound} and \eqref{eq:qaoaaddbound}).} We take $\ptNc  = \ptSBc = 1/2$ and $\piNc  = 0.999$, and these estimates become \hl{185} qubits for IQP circuits, \hl{370} qubits for QAOA circuits, and \hl{93} photons for boson sampling circuits. For these values of $\ptNc, \hl{\ptSBc}, \piNc $, the number of circuit elements needed for the lower bound to apply is $g_1(185) =$ 1,060,000 gates for IQP circuits, $g_2(370) =$ 2,110,000 constraints for QAOA circuits, and $g_3(93)=$ 17,400 beam splitters and phase shifters for boson sampling circuits. 

Thus, assuming one operation in the Word RAM model of computation corresponds to one floating-point operation on a supercomputer, and assuming our conjectures poly3-NSETH(1/2), per-int-NSETH(0.999), \hl{and poly3-ave-SBSETH(1/2)}, we conclude that classically simulating circuits of the sizes quoted above \hl{(up to either multiplicative or sufficiently small constant additive error)} would take at least a century on modern classical technology, a timespan we take to be sufficiently intractable.

If, additionally, we assume that the runtime of the classical simulation algorithm grows linearly with the number of circuit elements (e.g., the naive \hl{``Schr\"{o}dinger-style''} simulation algorithm that updates the state vector after each gate), then we can make a similar statement for circuits with many fewer gates. The cost of this reduction in gates is only a few additional qubits, due to the exponential scaling of the lower bound. We can estimate the number of qubits required by finding $q$ such that $s_i(q)/g_i(q) = \hl{10^{18}}\cdot 60\cdot 60 \cdot 24 \cdot 365/5$, the maximum number of supercomputer operations in \hl{ 1/500 of a century}, for $i=1,2,3\hl{,4,5}$.
We conclude that an IQP circuit with \hl{208 qubits and 500 gates}, a QAOA circuit with \hl{420 qubits and 500 constraints}, and a boson sampling circuit with \hl{98 photons and 500 linear optical elements} each would require at least one century --- one year per 5 circuit elements --- to be simulated using a classical simulation algorithm of this type running on state-of-the-art supercomputers. 

\hl{Here we remark again that to make these estimates, we have diverged from conventional complexity theory to make conjectures that fix even the constant prefactors on the lower bounds for algorithmic runtimes, and additionally, we must assert that these runtime lower bounds hold not only asymptotically but also in the $90 < n < 500$ range where our estimates fall. However, these estimates are robust in the sense that modifications to the constant prefactors have only minor effect on the final estimate. } 

We also note that the relative factor of two in the estimate for QAOA circuits is a direct consequence of the fact that one ancilla qubit was introduced per variable in order to implement the $H$ gates at the end of the IQP circuit $C_f$ within the QAOA framework. This illustrates how our estimate relies on finding a natural problem for these restricted models of quantum circuits and an efficient way to solve that problem within the model.  Indeed, an earlier iteration of this estimate based on the satisfiability problem instead of the degree-3 polynomial problem or matrix permanent required many ancilla qubits and led to a qubit estimate above 10,000. 
\hl{
\subsection{Relationship to Google's quantum computational supremacy experiment}

Since the appearance of the first version of this paper, a collaboration between Google and others reported that it has achieved quantum computational supremacy with 53 superconducting qubits on a programmable quantum device \cite{arute2019quantum}. Their interpretation and approach to quantum computational supremacy differs in certain respects from the perspective presented in our work. For one, the task they have completed on their quantum computer, Random Circuit Sampling (RCS) of quantum circuits with local gates on 2D lattices, is \textit{not} one of the proposals we have considered for our analysis. RCS does not fit nicely into our analysis since, while it is hard to classically perform RCS noiselessly (assuming non-collapse of $\PH$), RCS does not appear to correspond naturally with a \textit{specific} purely classical counting problem that has garnered independent study in classical computer science. 

Additionally, their quantum device is noisy and produces samples from a distribution that has neither small multiplicative nor small additive error with the ideal distribution. Rather, they argue that their device experiences global depolarizing noise and samples from a distribution that has small fidelity $F>0$ with the ideal distribution, i.e.~the device distribution (approximately) satsifes $P(x) = F Q(x) + (1-F)/2^n$, where $Q$ is the ideal distribution. When $F=1$ (noiseless), this task cannot be performed classically in polynomial time unless the $\PH$ collapses, and when $F=0$, $P$ is the uniform distribution and sampling from it is classically easy. They argue that the task is hard when $F$ is a small constant (in their experiment on the order of $10^{-3}$) by invoking what is essentially a fine-grained reduction in the same spirit as what we present here (see Supplementary Material of \cite{arute2019quantum}). They show how a time $T$ classical simulation drawing samples from $P(x)$ implies a time $O(FT)$ classical Arthur-Merlin ($\AM$) protocol that decides whether $Q(x)$ is greater than one threshold or less than a smaller threshold (assuming one is the case), which is essentially the RCS analogue of \texttt{poly3-SGAP}. As such, we will call this problem \texttt{RCS-SGAP}. Indeed, in light of our analysis, perhaps they could have opted to reduce to an $\mathsf{SB}$ algorithm, as we have, instead of an $\AM$ algorithm with a similar factor $O(F)$ overhead.

Thus, if they wished to derive a rigorous lower bound similar to the ones in our paper, they would need to impose a fine-grained conjecture asserting the non-existence of classical exponential-time $\AM$ protocols for \texttt{RCS-SGAP}. While it is certainly possible that such a conjecture could be true, the fact that it would not be about a classical computational problem would make it harder to begin to give evidence for it. The conjecture would essentially be an assertion that quantum circuits do not admit classical simulations that are substantially faster than brute-force, which would ideally be the conclusion of a QCS argument based off some unrelated conjecture instead of the conjecture itself.

However, it is apparent from their work that their interpretation of QCS is slightly different from ours. While we hope to rule out all possibility of competitive classical simulation algorithms, both known and unknown, they meticulously compare and benchmark their quantum device against best known classical algorithms for performing RCS, while putting less emphasis on arguing that yet-to-be-developed classical algorithms for RCS will not be substantially faster than ones that are currently known. On the one hand, this makes sense since ultimately QCS is about the \textit{practical} ascendancy of quantum computers over classical ones, and unknown classical algorithms are certainly not practical. On the other hand, one of the founding rationales \cite{aaronson2011computational} for sampling-based QCS was that it could be based on classical conjectures like the non-collapse of the $\PH$ whose refutation would bring cascading consequences to many other problems in classical complexity theory, as opposed to conjectures that merely assert that problems that are easy for quantum computers but not known to be easy for classical ones, such as factoring, in fact do not admit efficient classical algorithms. Under their more practically oriented interpretation, it is less important to base conjectures off of classical problems if one carefully studies best known existing algorithms for the quantum problem. 

In the end, they are able to declare QCS with only 53 qubits, a substantially smaller number than our estimates. This difference comes from our conservative approach to supercomputer performance, but more importantly from our aforementioned attempt to rule out all hypothetical simulation algorithms, known and unknown, and not just assume that the best known algorithms are optimal. Their effort to study the performance of best known classical simulation algorithms for RCS on one of the best supercomputers in the world reveals how supercomputer memory limitations are a key consideration that we have not introduced into our analysis. However, it is difficult for us to make a similarly detailed assessment of a supercomputer's performance when the classical algorithms we consider are hypothetical and the only thing we know about them is their runtime, forcing us to be conservative. Additionally, by insisting that we make conjectures about classical problems and not about the hardness of simulating the circuits themselves, we end up with simulation lower bounds that are \textit{not} tight with best known simulation algorithms, except in the case of boson sampling circuits. Our conjectures lead to qubit estimates that are larger by roughly a factor of 2 for IQP circuits and a factor of 4 for QAOA circuits, compared to if we had simply conjectured that best known simulation algorithms for those circuit families were optimal. 

Given our conservative approach and our more demanding requirements for QCS, it is in our view an encouraging sign that our qubit estimates are within a small constant factor of what the Google experiment was able to achieve.
}

\section{Conclusion}

Previous quantum computational supremacy arguments proved that polynomial-time simulation algorithms for certain kinds of quantum circuits would imply unexpected algorithms for classical counting problems within the polynomial-time hierarchy. We have taken this further by showing that even somewhat mild improvements over exponential-time best known simulation algorithms would imply non-trivial and unexpected algorithms for specific counting problems in certain cases. Thus, by conjecturing that these non-trivial classical counting algorithms cannot exist, we obtain lower bounds on the runtime of the simulation algorithms. In the case of boson sampling circuits, these lower bounds are essentially asymptotically tight when the strongest form of our conjecture is imposed.

Our conjectures for multiplicative-error simulation, poly3-NSETH($\ptNc $) and  per-int-NSETH($\piNc $), are fine-grained manifestations of the assumption that the $\PH$ does not collapse. \hl{Meanwhile, our conjecture for additive-error simulation, poly3-ave-SBSETH($\ptSBc$) is a fine-grained manifestation of the non-collapse assumption plus the statement that a certain counting problem is hard for the $\PH$ even on average.} While unproven, the non-collapse conjecture is extremely plausible; its refutation would entail many unexpected ramifications in complexity theory. This contrasts with the assumption that factoring has no efficient classical algorithm, which would also entail hardness of simulation but is less plausible because the consequences of its refutation on our current understanding of complexity theory would be minimal. Of course, the fine-grained nature of our conjectures makes them less plausible than the non-collapse of the $\PH$, but they are in line with current knowledge and beliefs in fine-grained complexity theory when $\ptNc, \hl{\ptSBc}  \leq 1/2$ and $\piNc  < 1$. 

It is worth comparing our approach with using a fine-grained version of the conjecture that $\PP\not\subset \Sigma_3^{\P}$, which is the complexity theoretic conjecture proposed in Aaronson-Arkhipov~\cite{aaronson2011computational}. An immediate advantage of our approach is that we avoid invoking Stockmeyer's theorem \cite{stockmeyer1983complexity} to estimate the acceptance probability of a hypothetical classical simulation algorithm for quantum circuits in $\Sigma_3P$; the fine-grained cost of this step would be significant, ultimately increasing our qubit estimates by roughly a factor of 3 \cite{dalzell2017lower}. Additionally, to understand the range of plausible fine-grained conjectures relating to $\Sigma_3$ algorithms, we might start with oracle bounds, analogous to our Theorem~\ref{thm:oracle-LB}.
 Known oracle lower bounds for the majority function show only that $\Sigma_3$ circuits that compute the majority of an oracle function (the oracle analogue of $\PP$) need size $\Omega(2^{n/5})$. This would correspond to taking $a$ or $b$ equal to $1/5$ which would increase our qubit estimates by a factor of $2.5$ or $5$ respectively.  The proof is also more complex, involving the switching lemma~\cite{dalzell2017lower}.  Thus our approach based on $\co\C_=\P \not\subset \NP$ instead of $\PP\not\subset \Sigma_3^{\P}$ yields both a much simpler proof and a tighter bound. 

The main motivation for imposing these fine-grained conjectures was to make an estimate of how large quantum circuits must be to rule out practical classical simulation on state-of-the-art classical computers. \hl{We chose the IQP, QAOA, and boson sampling models for our analysis because they are prominent QCS proposals and because their close connection with specific hard counting problems (i.e.~computing the gap of degree-3 polynomials over $\mathbb{F}_2$ and computing the permanent) made them amenable to perform a fine-grained analysis with little unnecessary overhead.} Our estimate relies on poly3-NSETH(1/2), per-int-NSETH(0.999), and \hl{poly3-ave-SBSETH(1/2)}, but it is somewhat robust to failure of these conjectures in the sense that if they fail in favor of mildly weaker versions, our estimate will increase only slightly. For example, replacing these conjectures with the slightly weaker poly3-NSETH($1/2d$), per-int-NSETH($1/d$), and \hl{poly3-ave-SBSETH($1/2d$)} increases the qubit estimate by only a factor of $d$, and replacing $2^{cn-1}$ time steps with $2^{cn-1}/d$ time steps in either conjecture (i.e.~$c\in \{\ptNc,\piNc\}$) increases the estimate by only $\log_2(d)$ qubits.

Our qubit estimates of fewer than 200 qubits for IQP circuits, fewer than 400 qubits for QAOA circuits, and fewer than 100 photons for boson sampling circuits are beyond current experimental capabilities but potentially within reach in the near future. Additionally, our estimate for boson sampling circuits is consistent with recently improved simulation algorithms \cite{neville2017classical,clifford2018classical} that can simulate circuits with up to as many as 50 photons but would quickly become intractable for higher numbers of photons.

\hl{Here we emphasize the importance of ruling out simulations with $O(1)$ additive error and not merely simulations with $O(1)$ multiplicative error. Experimental noise in real quantum systems without fault tolerance is likely to be large enough that most realistic devices could not achieve the noise rates for which our multiplicative-error bounds apply. This is not as problematic for the additive-error bounds; quantum systems with small but potentially reaslistic noise rates would generate an output distribution that has $\epsilon = O(1)$ additive error with respect to the ideal output distribution. Our bounds require that this additive error satisfy $\epsilon < 1/720$, and we argue that this could probably be improved to $\epsilon < 1/32$. These numbers are in line with previous additive-error QCS analyses (e.g., $\epsilon < 1/192$ in \cite{bremner2016average}), but finding ways to boost them might make QCS considerably more attainable for near-term devices under our analysis.

An important place our analysis can be improved is in providing additional evidence for our conjectures. While the conjectures poly3-NSETH($\ptNc $), per-int-NSETH($\piNc $), and poly3-ave-SBSETH($\ptSBc$) are consistent with other fine-grained conjectures like SETH, NSETH, and \#SETH, it is an open question whether it is possible to prove a concrete relationship with one of these conjectures, which would be important evidence in their favor. In particular, poly3-ave-SBSETH($\ptSBc$), which is a fine-grained statement about $\mathsf{SB}$ algorithms for an average-case problem, is unlike any conjecture previously proposed of which we are aware. The evidence we have provided, in the form of a black-box average-case query lower bound and a worst-case-to-quasi-average-case reduction, rules out certain methods one might use to refute it, but given little previous work on statements of this nature, it is hard for us to assess the likelihood that other methods for refuting it exist. 

Nevertheless, if anything, this highlights a weakness in QCS more generally. To arrive at practical qubit estimates, QCS arguments must be made fine-grained, but doing so uncovers the need to make conjectures about classical problems whose fine-grained complexity has not previously garnered considerable attention.  }

Finally, we conclude by noting that our analysis would likely be applicable to many other classes of quantum circuits whose efficient classical simulation entails the collapse of the $\PH$. \hl{Indeed, since the release of the first version of this paper, our method was extended in \cite{morimae2019fine} to derive fine-grained lower bounds for the simulation (up to multiplicative error) of the one-clean-qubit (DQC1) model, and the Hadamard-classical circuit model (HC1Q) \cite{morimae2018merlin}, as well as provide a lower bound on Clifford+$T$ simulation in terms of the number of $T$ gates. In \cite{morimae2019fine}, they also provide a compelling argument that it would be difficult to show that NSETH implies a version of the conjectures underlying their (and our) analysis. Additionally, in \cite{morimae2019depth}, the fine-grained lower bounds for simulation from our work and from \cite{monotone-sim-LB} (which used SETH to rule out exponential-time algorithms that compute output probabilities of quantum circuits) were extended to be based on other fine-grained assumptions, including the well-studied Orthogonal Vectors, 3-SUM, and All Pairs Shortest Path conjectures \cite{williams2015hardness}. } Other models that are universal under post-selection where this method may apply include various kinds of extended Clifford circuits \cite{jozsa2014classical,koh2015further}, and conjugated Clifford circuits \cite{bouland2017quantum}.

\hl{
\textit{Note:} The first version of this paper included only the lower bounds for multiplicative-error simulations and did not include any of the content of Section \ref{sec:adderror}. We would like to draw the reader's attention to Ref.~\cite{morimae2019additive} by Morimae and Tamaki, which was posted as the additive-error lower bounds that appear in the current version of this paper were in preparation. Ref.~\cite{morimae2019additive} is an independent analysis for additive-error QCS that shows several results, some of which overlap with ours in Section \ref{sec:adderror}, based on new fine-grained conjectures that are similar in spirit but different in detail to our poly3-ave-SBSETH($\ptSBc$). 
}

\section*{Acknowledgments}
We are grateful to Ashley Montanaro for important suggestions and conversations about degree-3 polynomials and the computational problem underlying this work. We would also like to thank Virginia Williams and Ryan Williams for useful discussions, and for pointing out Ref.~\cite{lokshtanov2017beating} and correcting an error in an earlier version of the paper. Finally, we acknowledge Richard Kueng for suggesting the proof strategy in Lemma \ref{lem:massbound}, as well as Tomoyuki Morimae and Suguru Tamaki for comments on a draft of this paper. 

AMD acknowledges support from the Undergraduate Research Opportunities Program (UROP) at MIT, the Dominic Orr Fellowship at Caltech, and the National Science Foundation Graduate Research Fellowship under Grant No. DGE‐1745301. 
AWH was funded by NSF grants CCF-1452616 and CCF-1729369, ARO contract
W911NF-17-1-0433 and the MIT-IBM Watson AI Lab under the project  {\it Machine Learning in Hilbert space}.
DEK was supported
by the National Science Scholarship from the Agency for Science, Technology and Research (A*STAR) and Enabling Practical-scale Quantum Computation (EPiQC), a National Science Foundation (NSF) Expedition in Computing, under grant CCF-1729369.
RLP was funded by the MIT School of Science Fellowship, the MIT Department of Physics, and the MIT-IBM Watson AI Lab.

\appendix

\section{Reduction from \texttt{poly3-NON-} \texttt{BALANCED} to \texttt{per-int-NONZERO}}\label{app:coC=Phard}

Valiant famously showed that computing the permanent of an integer matrix is \#\P-hard by reduction from \texttt{\#3SAT} \cite{valiant1979complexity}. A concise reproduction of this proof can be found in \cite{arora2009computational}. The main idea for our reduction is the same, the only change being in the details of the clause and variable gadgets we use for the construction.

There is a bijective correspondence between $n \times n$ matrices and directed graphs with $n$ vertices, where the entry $A_{ij}$ of a matrix $A$ corresponds to the edge weight from vertex $i$ to vertex $j$ in the associated graph $G_A$. A cycle cover of $G_A$ is a subset of the edges of $G_A$ forming some number of cycles in which each vertex appears in exactly one cycle. The weight of a cycle cover is the product of the weights of all the edges traversed by one of the cycles. From the definition of the permanent in Eq.~\eqref{eq:permanent}, we can see that the sum of the weights of all the cycle covers of $G_A$ is given by $\Per(A)$. 

It will be straightforward to convert the reduction from \texttt{\#3SAT} to computing the permanent into a reduction from \texttt{poly3-NONBALANCED} to \texttt{per-int-NONZERO} since degree-3 polynomials and 3-CNF formulas have a common structure in the sense that both involve $n$ variables where groups of three variables appear together in terms/clauses.

Suppose we are given a degree-3 polynomial $f$ with $n$ variables and $m$ clauses. We build a corresponding graph $G_f$ by including one term gadget for each of the $m$ terms and one variable gadget for each of the $n$ variables, and then connecting them in a certain way. These gadgets are shown in Figure \ref{fig:termgadget} and Figure \ref{fig:variablegadget}. If a term of $f$ has fewer than three variables, we can repeat one of the variables that appears in that term (e.g.~$x_1x_2 = x_1x_2x_2$), and thereby assume that each term has three variables. Each term gadget has three dotted edges corresponding to the three variables that appear in that term. A variable that appears $t$ times will have $t$ dotted edges in its variable gadget. Thus, each dotted variable edge from some node $u$ to node $u'$ has a corresponding dotted edge from node $v$ to node $v'$ in a term gadget associated with a term in which that variable appears. Each such pair of dotted edges indicates that the nodes $u$, $u'$, $v$ and $v'$ should be connected using the XOR gadget shown in Figure \ref{fig:xorgadget}. Thus, the dotted edges are not part of the final graph. The XOR gadget has the effect of ensuring that any cycle cover of the graph uses one of the two dotted edges but not both. The effective weight of an edge connected to an XOR gadget is 4. 

\begin{figure}[ht]
\begin{center}
\begin{tikzpicture}[->,>=stealth',shorten >=1pt,auto,node distance=3cm,thick,main node/.style={circle,draw,font=\sffamily\Large\bfseries,minimum size=24 pt}]

  \node[main node] (2)  {};
  \node[main node] (1) [above of=2]{v'};
  \node[main node] (3) [below right of=2] {v};
  \node[main node] (4) [below left of=2] {};

  \path[every node/.style={font=\sffamily\small}]
    (1) edge [bend right=10] node[left] {} (2)
    	edge [dashed,bend right=25] node[left] {4} (4)
    (2) edge [bend right=10] node [right] {} (1)
        edge [bend right=10] node {} (4)
        edge [loop left] node {-1} (2)
        edge [bend right=10] node[left] {} (3)
    (3) edge [bend right=10] node[right] {} (2)
        edge [bend right=10] node[right] {} (4)
        edge [dashed,bend right=25] node[right] {4} (1)
    (4) edge [bend right=10] node [right] {} (3)
        edge [bend right=10] node[right] {2} (2)
        edge [dashed,bend right=45] node[below] {4} (3);
\end{tikzpicture}
\end{center}

\caption{\label{fig:termgadget} Gadget for each term in the degree-3 polynomial $f$. Unlabeled edges are assumed to have weight 1. The three dashed lines are connected via the XOR gadget to the dashed lines in the variable gadgets for the variables that appear in the term, as exemplified by the labeling of vertices $v$ and $v'$ in the context of Figure \ref{fig:xorgadget}. If all three variables are true, the term gadget will contribute a cycle cover factor of $-1$, excluding the factors of 4 from dotted edges. If at least one variable is false, the term will contribute a cycle cover factor of 1.}
\end{figure}
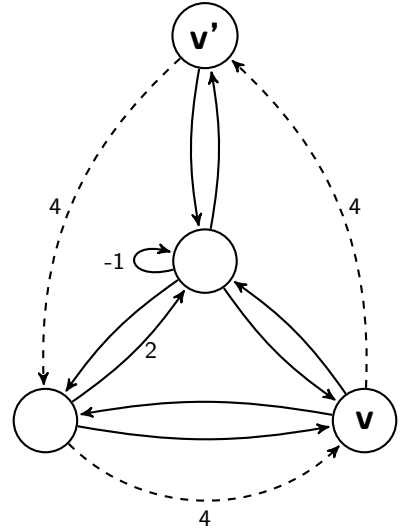

\begin{figure}
\begin{center}
\begin{tikzpicture}[->,>=stealth',shorten >=1pt,auto,node distance=3cm,
                    thick,main node/.style={circle,draw,font=\sffamily\Large\bfseries,minimum size=24 pt}]

  \node[main node] (1) {};
  \node[main node] (2) [below of=1] {};
  \node[main node] (3) [right of=2] {u};
  \node[main node] (4) [right of=3] {u'};
  \node[main node] (5) [above of=4] {};

  \path[every node/.style={font=\sffamily\small}]
    (1) edge [bend left=20] node[left] {} (5)
    	edge [dashed,bend right=25] node[left] {4} (2)
    (2) edge [dashed,bend right=25] node [below] {4} (3)
        edge [loop below] node {} (2)
    (3) edge [dashed,bend right=25] node [below] {4} (4)
        edge [loop below] node {} (3)
    (4) edge [dashed,bend right=25] node [right] {4} (5)
        edge [loop below] node {} (4)
    (5)	edge node[right] {} (1) ;
\end{tikzpicture}
\end{center}

\caption{\label{fig:variablegadget} Gadget for each variable in the degree-3 polynomial $f$. The number of dashed lines is equal to the number of terms in which the variable appears, so this example is for a variable that appears in four terms. The dashed lines are connected to the dashed lines in the term gadget in which that variable appears via the XOR gadget, as exemplified by the labeling of vertices $u$ and $u'$ in the context of Figure \ref{fig:xorgadget}.}
\end{figure}
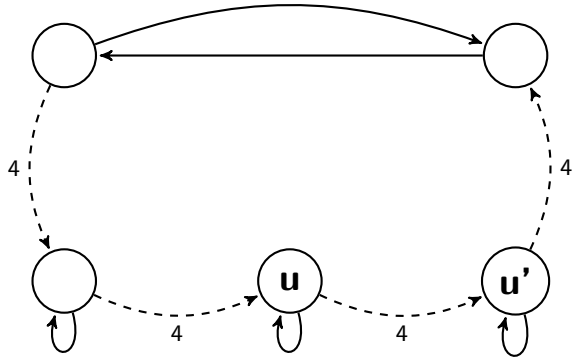

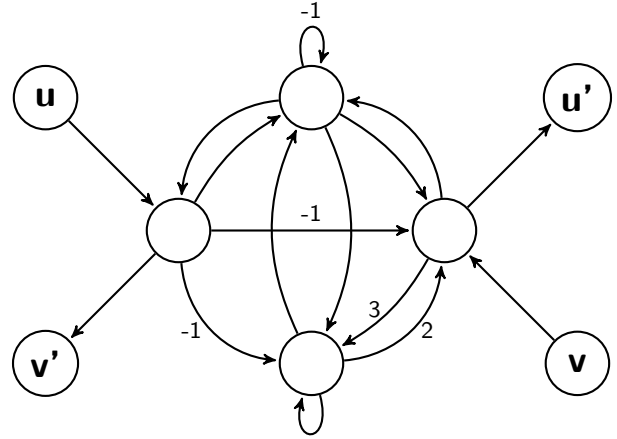
\begin{figure}[ht]
\begin{center}
\begin{tikzpicture}[->,>=stealth',shorten >=1pt,auto,node distance=2.5cm,
                    thick,main node/.style={circle,draw,font=\sffamily\Large\bfseries,minimum size=24 pt}]

  \node[main node] (1) {u};
  \node[main node] (2) [below right of=1] {};
  \node[main node] (3) [below left of=2] {v'};
  \node[main node] (4) [above right of=2] {};
  \node[main node] (5) [below right of=2] {};
  \node[main node] (6) [below right of=4] {};
  \node[main node] (7) [below right of=6] {v};
  \node[main node] (8) [above right of=6] {u'};

  \path[every node/.style={font=\sffamily\small}]
    (1) edge node[left] {} (2)
    (2) edge node[left] {} (3)
    	edge [bend left=15] node[left] {} (4)
        edge [bend right=40] node[left] {-1} (5)
        edge  node[above] {-1} (6)
    (4) edge [bend right=40] node[left] {} (2)
    	edge [bend left=15] node[left] {} (6)
        edge [bend left=25] node[left] {} (5)
        edge [loop above] node {-1} (4)
    (5) edge [bend right=40] node[right] {2} (6)
    	edge [bend left=25] node[left] {} (4)
        edge [loop below] node {} (4)
    (6) edge node[left] {} (8)
        edge [bend right=40] node[left] {} (4)
        edge [bend left=15] node[left] {3} (5)
    (7) edge node[left] {} (6);
\end{tikzpicture}
\end{center}

\caption{\label{fig:xorgadget} XOR gadget that connects dotted lines from node $u$ to $u'$ in the variable gadget with dotted lines from node $v$ to $v'$ in the term gadget. The effect of the XOR gadget is that any cycle cover must use either the edge from $u$ to $u'$ or the edge from $v$ to $v'$, but not both. Each XOR gadget contributes a factor of 4 to the weight of the cycle cover.}
\end{figure}

Every cycle cover of $G_f$ corresponds to some setting of the variables $z_1, \ldots, z_n$. If the cycle cover traverses the solid lines at the top of the variable gadget associated with variable $z_j$, then the corresponding setting has $z_j=1$. In this case, the cycle cover cannot also traverse the dotted lines at the bottom of the $z_j$ variable gadget. Thus, due to the XOR gadget, the cycle cover \textit{must} traverse the dotted lines corresponding to $z_j$ in each term gadget associated with a term in which $z_j$ appears. 

On the other hand, if the cycle cover uses the dotted lines in the $z_j$ gadget instead of the solid lines at the top, this corresponds to $z_j=0$, and the cycle cover cannot also traverse the edges corresponding to $z_j$ in the term gadgets associated with terms in which $z_j$ appears.

When all three dotted edges of a term gadget are traversed, this corresponds to all three variables in the term being set to 1. There is only one way to cycle cover the term gadget in this case, and it has a weight of $-1$, excluding the factors of 4 that come from the dotted edges in the XOR gadget. Meanwhile, if at least one dotted edge in the term gadget is not traversed, the total weight of all cycle covers will contribute a factor of $1$, again excluding the factors of 4. Thus, each assignment $z$ for which $f(z) = 0$ corresponds to cycle covers that satisfy an even number of terms, with total weight $4^{3m}$ since exactly $3m$ XOR gadgets are involved. Each assignment for which $f(z)=1$ corresponds to cycle covers that satisfy an odd number of terms, with total weight $-4^{3m}$. Thus, the total cycle cover weight of $G_f$, and by extension the permanent of the integer-valued matrix corresponding to $G_f$ is non-zero if and only if $\gap(f) \neq 0$. The number of vertices in $G_f$ is a polynomial in the number of variables of $f$, so this completes the reduction from \texttt{poly3-NONBALANCED} to \texttt{per-int-NONZERO}. Since \texttt{poly3-NONBALANCED} is $\co\C_=\P$-complete, \texttt{per-int-NONZERO} is $\co\C_=\P$-complete as well.

\section{Better-than-brute-force solution to \texttt{poly3-NONBALANCED}}\label{app:algo}

LPTWY~\cite{lokshtanov2017beating} gave a better-than-brute-force randomized algorithm that determines whether a system of $m$ degree-$k$ polynomial equations over finite field $\mathbb{F}_q$ has a solution (i.e.~a setting of the variables that makes all $m$ polynomials equal to 0). They also derandomized this procedure to create a better-than-brute-force deterministic algorithm that \textit{counts} the number of solutions to a system of $m$ degree-$k$ polynomial equations over finite field $\mathbb{F}_q$. Applying their deterministic algorithm for the special case $m=1$, $k=3$, $q=2$ (for which it is considerably simpler) yields a deterministic solution for \texttt{poly3-NONBALANCED}. We give a simple reproduction of their algorithm in this case below.

\begin{theorem}
There is a deterministic algorithm for \texttt{poly3-NONBALANCED} running in time $\text{poly}(n)2^{(1-\delta) n}$ where $\delta = 0.0035$.
\end{theorem}
\begin{proof}
The algorithm beats brute force by finding a clever way to efficiently represent the number of zeros of a degree-3 polynomial with $n$ variables when $(1-\delta)n$ of the variables have been fixed. Then, by summing the number of zeros associated with the $2^{(1-\delta)n}$ possible settings of these variables, the algorithm computes the total number of zeros in $\text{poly}(n)2^{(1-\delta)n}$ time, which is better than brute-force $\text{poly}(n)2^n$.

First we describe the algorithm. The input is the degree-3 polynomial $f$, which has $n$ variables. In the following we have $x \in \{0,1\}^n$, and we let $y$ be the first $(1-\delta)n$ bits of $x$ and $a$ be the last $\delta n$ bits of $x$. Following the notation from \cite{lokshtanov2017beating}, we define

\begin{equation}\label{eq:Qlya}
\hat{Q}_l(y,a) = 1-(1-f(x))^l \sum_{j=0}^{l-1} \binom{l+j-1}{j} f(x)^j.
\end{equation}

In \cite{beigel1994acc}, it is shown that if $f(x) \equiv 0 \mod 2$, then $\hat{Q}_l(y,a) \equiv 0 \mod 2^l$ and if $f(x) \equiv 1\mod 2$, then $\hat{Q}_l(y,a) \equiv 1\mod 2^l$. We define

\begin{equation}\label{eq:rly}
R_l(y) = \sum_{a \in \{0,1\}^{\delta n}}  \hat{Q}_l(y,a)
\end{equation}
and observe that $R_l(y)$ gives the number of settings $x$ (mod $2^l$) for which $f(x)=1$ and the first $(1-\delta)n$ bits of $x$ are $y$. 

The algorithm operates by enumerating all values of $y$, computing $R_l(y)$ when $l = \delta n$ (which is large enough so that the number of settings for which $f(x)=1$ will never exceed $2^l$ for a given value of $y$), and summing all the results. This gives the total number of inputs $x$ for which $f(x) = 1$. The algorithm rejects if this number is $2^{n-1}$, and otherwise accepts.

There are two contributions to the runtime. The first is the computation of a representation of $R_{\delta n}(y)$ as a sum of monomials in $(1-\delta) n$ variables of $y$ with integer coefficients. Each monomial has degree at most $6 \delta n -3$. The number of possible monomials with coefficient 1 over $a$ variables with degree at most $b$ is
\begin{equation}\label{eq:nummonomials}
M(a,b) = \binom{a+b}{b} \leq (1+a/b)^b(1+b/a)^a,
\end{equation}
and, from Eq.~\eqref{eq:Qlya}, it is apparent that $\hat{Q}_{\delta n}(y,a)$ can be computed by a polynomially long sequence of sums or products of a pair of polynomials, where a product always includes either the polynomial $(1-f(x))$ or $f(x)$, which have degree only 3. Thus each step in the sequence takes time at most $\text{poly}(n)M((1-\delta)n, 6\delta n-3)$. For a certain value of $a$, a polynomial number of such steps required to create a representation of $\hat{Q}_{\delta n}(y,a)$ and then $R_{\delta n}$ is the sum over $2^{\delta n}$ such representations (one for each setting of $a$). Thus the total time is also bounded by $\text{poly}(n)2^{\delta n} M((1-\delta)n, 6\delta n-3)$. 

The second contribution to the runtime is the evaluation of this polynomial for all points $y$, given its representation computed as described. It is shown in Lemma 2.3 of \cite{lokshtanov2017beating} that this evaluation can be performed in time $\text{poly}(n)2^{(1-\delta)n}$, so long as the representation of $R_{\delta n}$ has fewer than $2^{0.15(1-\delta)n}$ monomials. This is satisfied as long as
\begin{equation}
M((1-\delta)n, 6 \delta n -3) \leq 2^{0.15(1-\delta)n},
\end{equation}
which, using Eq.~\eqref{eq:nummonomials}, can be seen to occur whenever $\delta < 0.0035$. This is an improvement on the general formula in \cite{lokshtanov2017beating}, which when evaluated for $k=3$ and $q=2$ yields a bound of $\delta < 0.00061$. 

Assuming $\delta$ satisfies this bound, the total runtime is the sum of the two contributions, $\text{poly}(n)2^{\delta n} M((1-\delta)n, 6\delta n-3)+ \text{poly}(n)2^{(1-\delta)n}$. The first term is smaller than $\text{poly}(n)2^{(0.85\delta+0.15)n}$, so the second term dominates, and the total runtime is $\text{poly}(n)2^{(1-\delta)n}$, proving the theorem.
\end{proof}

Consider ways in which the runtime could be improved. Suppose the evaluation time were to be improved such that the polynomial $R_{\delta n}$ could be evaluated in time $\text{poly}(n)2^{(1-\delta)n}$ even when $\delta >0.5$. With no further changes to the algorithm, the first contribution to the runtime stemming from the time required to compute the representation of $R_{\delta n}$ would now dominate and the runtime would still exceed $2^{0.5 n}$. Moreover, as long as $R_l$ is expressed as a sum over $2^{\delta n}$ terms as in Eq.~\eqref{eq:rly}, it is hard to see how any current techniques would allow this representation to be computed in less than $2^{0.5n}$ time when $\delta > 0.5$. 

Stated another way, this method of beating brute force by enumerating over only a fraction $(1-\delta) n$ of the variables and evaluating the number of solutions when those variables have been fixed in $2^{(1-\delta)n}$ time will surely break down when $\delta > 0.5$ because there will be more variables not fixed than fixed, and the preparation of the efficient representation of the number of zeros will become the slowest step.

\section{The moments of the \texorpdfstring{$\gap(f)$}{gap(f)} distribution}\label{app:moments}

In this appendix we compute the limiting value of the moments of the distribution of the quantity $(\gap(f)/2^n)^2$ where $f$ is drawn uniformly at random from the set of degree-3 polynomials with $n$ variables over the field $\mathbb{F}_2$. As the number of variables increases, the values become arbitrarily close to what they would be if $\gap(f)/2^n$ were a random Gaussian variable with mean 0 and variance $2^{-n}$. We believe this observation could be relevant in other applications. Following the proof, we briefly comment on other possible implications and also what would be different if $f$ is only a uniformly random degree-2 polynomial. Then, we use this fact to show that the amount of probability mass in the NO part of the distribution (see Figure \ref{fig:gapfdist}) is at least 0.12 for sufficiently large $n$.

Let $\mathcal{F}_n$ be the set of degree-3 polynomials over the field $\mathbb{F}_2$ with $n$ variables and no constant term and let $\ngap(f) = \gap(f)/2^n$.

\begin{theorem}\label{thm:moments}
For any $k$ and any $\epsilon$, there exists a constant $n_0$ such that whenever $n > n_0$
\begin{equation}
    \lvert 2^{nk}\mathbb{E}_{f \in \mathcal{F}_n} (\ngap(f)^{2k}) - (2k-1)!! \rvert \leq \epsilon,
\end{equation}
where $(2k-1)!! = 1*3*5*\ldots*(2k-1)$
\end{theorem}

\begin{proof}

A degree-3 polynomial $f$ over $\mathbb{F}_2$ with no constant term can be written
\begin{equation} 
f(z_1, \ldots, z_n) = \sum_{a,b,c=1}^n \alpha_{abc} z_az_bz_c,
\end{equation}
where $\alpha_{ijk} \in \{0,1\}$. Note that $z_p = z_p^2 = z_p^3$ when $z_p \in \mathbb{F}_2$. This is not a one-to-one mapping. Degree-1 monomials $z_p$ correspond to the single term in the expansion for which which $a=b=c=p$, but degree-2 monomials $z_pz_q$ correspond to the six terms in which, for example, $a=b=p$ and $c=q$. Degree-3 monomials also each correspond to six terms in the expansion. Using this correspondence
\begin{eqnarray}
    \ngap(f) &=& 2^{-n}\sum_x (-1)^{f(x)} \nonumber \\
    &=& 2^{-n} \sum_x (-1)^{\sum_{a,b,c} \alpha_{abc} x_a x_b x_c} \nonumber \\
    &=& 2^{-n} \sum_x \prod_{a,b,c} (-1)^{\alpha_{abc} x_a x_b x_c}.
\end{eqnarray}

Choosing the coefficients $\alpha_{abc}$ uniformly at random from $\{0,1\}$ is equivalent to choosing a degree-3 polynomial from $\mathcal{F}_n$ uniformly at random. Thus we may express

\begin{eqnarray}\label{eq:momentsum}
   && 2^{2nk}\mathbb{E}_{f \in \mathcal{F}_n}(\ngap(f)^{2k})  \\
    &=& \mathbb{E}_{\alpha}\left[\sum_{x^1,\ldots, x^{2k}} \prod_{a,b,c = 0}^n (-1)^{\alpha_{abc}(x_a^1x_b^1x_c^1+\ldots + x_a^{2k}x_b^{2k}x_c^{2k})}\right] \nonumber \\
    &=& \sum_{x^1,\ldots, x^{2k}} \prod_{a,b,c = 0}^n \mathbb{E}_{\alpha_{abc}}\left[(-1)^{\alpha_{abc}(x_a^1x_b^1 x_c^1+\ldots + x_a^{2k}x_b^{2k}x_c^{2k})}\right] \nonumber \\
    &=& \sum_{x^1,\ldots, x^{2k}} \prod_{a,b,c = 0}^n \mathds{1}\left(x_a^1x_b^1 x_c^1+\ldots + x_a^{2k}x_b^{2k}x_c^{2k} = 0\right), \nonumber
\end{eqnarray}
where $\mathds{1}$ is the indicator function.

For each of the $2^{2nk}$ terms in the sum above, there is a corresponding $2k \times n$ matrix $X = (x^j_a)$ over $\mathbb{F}_2$. This term contributes 1 to the overall expectation if for any three columns $X_a$, $X_b$, $X_c$ (each vectors with $2k$ entries), $\langle X_a,X_b,X_c \rangle = 0$, where
\begin{equation}
    \langle u,v,w \rangle = \sum_{j=1}^{2k} u^j v^j w^j \mod 2
\end{equation}
is an extension of the standard inner product $\langle u, v \rangle = \sum_j u^jv^j \mod 2$ to three vectors. Otherwise, the term yields 0. Thus, the problem amounts to counting the number of $2k \times n$ matrices $X$ for which this condition is met. In what follows, we will perform this counting by showing that the number is dominated by matrices $X$ whose $2k$ rows ``pair up'' into $k$ sets of 2, where rows are identical to their pair. This is reminiscent of the Wick contractions used to compute Gaussian integrals and is fundamentally why the moments agree with those of a Gaussian in the limit $n \rightarrow \infty$.

To count the number of terms $X$ that yield 1, we use the properties of $\mathbb{F}_2^{2k}$ as a $2k$-dimensional vector space with inner product $\langle \cdot, \cdot \rangle$. We associate with each matrix $X$, a subspace $H_X \subset \mathbb{F}_2^{2k}$ formed by taking linear combinations of the columns of $X$. Another way to say this is that $H_X$ is the linear binary code generated by the transpose of $X$. For any subspace $V$ we let $V^\perp$ contain vectors that are orthogonal to all the elements of $V$. Note that a vector may be orthogonal to itself. We also define the operation $\times$ to be entry-wise vector multiplication between two vectors. This multiplication operation turns $\mathbb{F}_2^{2k}$ into an algebra. For any subspace $V$, we let $V^\times$ denote the subspace generated by pairwise products $u \times v$ where $u,v$ are in $V$. We will be interested in subspaces $H$ that satisfy the condition
\begin{equation}\label{eq:condition}
    H^\times \subset H^\perp.
\end{equation}
We claim that a term $X$ in the sum above yields 1 if and only if $H_X$ satisfies condition \eqref{eq:condition}. This is seen as follows. 

If $u \in H_X$ and $v \in H_X^\times$, then $\langle u, v \rangle$ may be decomposed into a sum of $\langle X_a, X_b \times X_c \rangle$ for various $a,b,c$ by writing $u$ and $v$ as a linear combination of columns of $X$ and of products of columns of $X$, respectively. If the term yields 1, then for any three columns $X_a, X_b, X_c$ of $X$, $\langle X_a, X_b, X_c \rangle = \langle X_a, X_b \times X_c \rangle = 0$, implying that $H_X^\times$ is orthogonal to $H_X$. Conversely, we have $X_a \in H_X$ and $X_b \times X_c \in H_X^\times$ for all $a,b,c$ so if the condition \eqref{eq:condition} holds, then $\langle X_a, X_b \times X_c \rangle = \langle X_a, X_b, X_c \rangle =0$, implying the term $X$ yields 1.

Given a subspace $H$ of dimension $d$ satisfying condition $\eqref{eq:condition}$, for how many matrices $X$ does $H=H_X$? All such $X$ can be formed by choosing $n$ columns among the $2^d$ elements in $H$, giving an initial count of $2^{dn}$. However, for some of these matrices, $H_X$ will merely be a (strict) subset of $H$. The number of these matrices is at most $2^{n(d-1)}2^d$ since there are at most $2^d$ strict subspaces of $H$ with dimension $d-1$. In the limit of large $n$, this correction will vanish compared to $2^{dn}$. 

This establishes that the exponential growth of the number of matrices $X$ for which $H_X = H$ is governed by the dimension $d$ of $H$, so to leading order we must merely calculate how many distinct $H$ have maximal dimension. We claim that the maximum dimension is $d=k$ and the number of $H$ with this dimension is given by the number of distinct ways to partition $2k$ indices into $k$ pairs. We see this as follows.

Assume $H$ satisfies the condition. Note that the dimension of $H^\perp$ is $2k-d$, and since $H \subset H^\times \subset H^\perp$, $d \leq 2k-d$, and hence $d \leq k$.  When $d=k$, $H = H^\times = H^\perp$, which implies a couple of important facts. First, the all ones vector $\mathbf{1}$ must be in $H$, since $\langle \textbf{1}, u \rangle = \langle u, u \rangle = 0$ and hence $\mathbf{1} \in H^\perp=H$. Second, the space $H$ is a subalgebra of $\mathbb{F}_2^{2k}$, since the fact that $H=H^\times$ implies it is closed under multiplication. We may choose a basis $v^1, \ldots, v^k$ for $H$, with $v^1 = \mathbf{1}$. For any fixed index $a$, we may additionally require that $v^j_a=1$ (the $a$th entry of $v^j$) for all $j$, because we may always add $\mathbf{1}$ to a basis vector if $v^j_a=0$. Since $H$ is a subalgebra, $v^1 \times \ldots \times v^k \in H$, and by construction it has a 1 in its $a$th entry. Moreover, $0=\langle \mathbf{1}, v^1 \times \ldots \times v^k \rangle$ since  $H\subset H^\perp$. Therefore $v^1 \times \ldots \times v^k$ must also have a 1 in some other entry $b \neq a$, and hence $v^j_b = 1$ for all $j$. Since there is a basis for $H$ in which all the basis vectors have $1$s in both index $a$ and index $b$, every vector in $H$ has the same value in index $a$ and index $b$. This shows that if $H$ has dimension $k$ and satisfies condition \eqref{eq:condition}, then every index $1 \leq a \leq 2k$ must have a partner $b \neq a$ for which $u_a = u_b$ for every $u \in H$. Conversely, if every index pairs up in this sense, then $H$ must satisfy condition \eqref{eq:condition}, since each pair of indices always contributes a pair of 0s or a pair of 1s to the sum $\langle u, v \times w \rangle = \sum_a u_a v_a w_a \mod 2$. 

Thus the number of distinct $k$-dimensional $H$ satisfying condition $\eqref{eq:condition}$ is the number of ways to pair up the $2k$ indices, given by the expression
\begin{equation}
    (2k-1)*(2k-3)* \ldots *7*5*3*1 = (2k-1)!!
\end{equation}
For each of these $H$, there are $2^{kn}$ $X$ for which $H_X = H$, with corrections of at most $c_k 2^{(k-1)n}$ with $c_k$ depending only on $k$. In addition, there may be $H$ with dimension $d < k$ that satisfy condition \eqref{eq:condition}, but for each of these the number of $X$ for which $H_X=H$ is at most $2^{dn} \leq 2^{(k-1)n}$. Moreover, the number of subspaces $H$ with dimension $d < k$ is given by some number $c'_k$ depending only on $k$. Thus, we may choose $n_0$ large enough so that $2^{-n_0}c_k*(2k-1)*\ldots*5*3*1 \leq \epsilon/2$ and $2^{-n_0}c'_k \leq \epsilon/2$. 

Recalling the $2^{2nk}$ prefactor in Eq.~\eqref{eq:momentsum}, this proves the statement of the theorem. 
\end{proof}

If $f$ were only a degree-2 polynomial, this theorem would not hold. In particular, terms in the sum associated with matrix $X$ would yield 1 if and only if all pairs of columns $X_a,X_b$ of $X$ satisfy $\langle X_a, X_b \rangle = 0$. It need not also be true that $\langle X_a, X_b, X_c \rangle = 0$ for any trio of columns. Thus the condition \eqref{eq:condition} is replaced by the less restrictive statement $H \subset H^\perp$. Note that this condition is precisely the statement that the linear binary code $H$ is contained in its dual. It is still the case that the maximum dimension of any such $H$ is $d=k$, in which case $H = H^\perp$ (i.e.~$H$ is self-dual), but there are a larger number of subspaces $H$ with dimension $k$ that satisfy the new condition, since it is possible that $H=H^\perp$, while $H^\times$ forms a larger subspace. For example, we can take $H$ be generated by the columns of
\begin{equation}
    \begin{pmatrix}
        1 & 0 & 0 & 0 \\
        1 & 0 & 0 & 1 \\
        1 & 0 & 1 & 0 \\
        1 & 0 & 1 & 1 \\
        1 & 1 & 0 & 0 \\
        1 & 1 & 0 & 1 \\
        1 & 1 & 1 & 0 \\
        1 & 1 & 1 & 1 
    \end{pmatrix}.
\end{equation}
Here, none of the rows are identical so none of the indices have paired up. As expected, $H^\times \not\subset H^\perp$ as seen by the fact that the inner product of the second column with the entrywise product of the third and fourth columns is 1. Yet it is still true that $H \subset H^\perp$.

We showed above that the number of subspaces $H$ that satisfy condition \eqref{eq:condition} and have dimension $k$ is $(2k-1)!!$. There are additional matrices that satisfy just the weaker degree-2 condition $H = H^\perp$, bringing the total to (whenever $k >1$)
\begin{eqnarray}
&&\frac{(2^{2k-1}-2^1)(2^{2k-2}-2^2)\ldots(2^{k+1}-2^{k-1})}{(2^k-2^1)(2^k-2^2)\ldots(2^k-2^{k-1})} \\
&=&(2^1+1)(2^2+1)\ldots(2^{k-1}+1),
\end{eqnarray}
which is derived by first choosing $k-1$ vectors that are linearly independent to form a basis (along with the all ones vector $\mathbf{1}$) for $H$ (numerator), and then dividing by the number of bases associated with a particular subspace. 

For $k=1,2,3$, this number is equal to $(2k-1)!!$, so the first three moments agree with a Gaussian even for degree-2 polynomials. However, for $k=4$ this expression evaluates to 135, while $7!! = 105$. Indeed, as $k$ increases, the expression grows much more quickly, like $2^{O(k^2)}$, than $(2k-1)!!$ grows. 

The statement that the moments of the quantity $\gap(f)$ match those of a Gaussian when $f$ is a degree-3 polynomial could have consequences beyond the scope of this work. It might be possible to use this analysis to formally show that the $\gap(f)$ distribution itself approaches a Gaussian. It is also important to note that if $f$ is instead drawn uniformly at random from the set of all Boolean functions, the distribution of $\gap(f)$ would be exactly a Binomial distribution, which has the same moments for large $n$. Thus, in a sense, random degree-3 polynomials might be thought to be mimicking the behavior of completely random Boolean functions (with many fewer parameters), and the same cannot be said for degree-2 polynomials. Perhaps this could be connected to the fact that computing $\gap(f)$ for degree-3 polynomials is $\#P$-hard, while doing the same for degree-2 polynomials is easy.

Next we prove a statement about the probability mass near 0 in the distribution for degree-3 polynomials.

\begin{lemma}\label{lem:massbound}
    There exists an $n_0$ such that for all $n > n_0$
    \begin{equation}
        \Pr_{f \in \mathcal{F}_n}\left[ (\gap(f)/2^n)^2 \leq 2^{-n-2}\right] \geq 0.12.
    \end{equation}
\end{lemma}
\begin{proof}
Let $I(x)$ be the indicator function that evaluates to 1 when $x \leq 1/4$ and to 0 otherwise. Suppose $p(x)$ is a polynomial of degree $L$ for which $p(x) \leq I(x)$ whenever $x \geq 0$. By Theorem \ref{thm:moments}, the first $L$ moments $\mathbb{E}_{f \in \mathcal{F}_n} ( 2^{kn} \ngap(f)^{2k})$ can be made arbitrarily close to their Gaussian values by taking $n_0$ large. Thus for any $\epsilon$, we may take $n_0$ large enough so that
\begin{eqnarray}
&&\Pr_{f \in \mathcal{F}_n}\left[ (\gap(f)/2^n)^2 \leq 2^{-n-2}\right] \nonumber \\
&=& \mathbb{E}_{f \in \mathcal{F}_n} \left[I( 2^n \ngap(f)^2)  \right] \nonumber \\
&\geq& \mathbb{E}_{f \in \mathcal{F}_n} \left[p( 2^n \ngap(f)^2)  \right] \nonumber \\
&\geq& \mathbb{E}_{x \sim \mathcal{N}(0,1)} (p(x^2)) -\epsilon,
\end{eqnarray}
where $\mathcal{N}(0,1)$ is the Gaussian distribution with mean 0 and variance 1. 

We construct a specific polynomial $p(x)$ with degree $L=15$.
\begin{equation}
    p(x) = \frac{\delta^2}{1-\delta^2} \left(T_L(\sqrt{1+A-4Ax})^2-1 \right),
\end{equation}
where $T_L$ is the $L$th Chebyshev polynomial of the first kind, $\delta = 0.5$ and $A = T_{1/L}(1/\delta)^2-1 = 0.00773$. 

It can be verified that for any odd choice of $L$ and any choice of $\delta$, this polynomial is smaller than $I(x)$ for all $x \geq 0$. We can also give the coefficients explicitly by writing
\begin{equation}
    p(x) = \sum_{j=0}^{15} \frac{c_j}{(2j-1)!!}x^j,
\end{equation}
where $\{c_j\}_{j=0}^{15}$ is given by

\begin{equation}
    \begin{split}
        \{  &  1,\; -6.0672, \; 29.9730, \; -114.8688, \\
        & 345.0021, \; -829.2997,\;  1620.0455, \; -2593.7392, \\
        & 3410.0118, \; -3665.1216,\; 3183.4033,\;  -2188.3186, \\
        & 1149.8164, \; -435.1008,\;  105.8449,\; -12.4590 \}.
    \end{split}
\end{equation}

Then, since $\mathbb{E}_{x \sim \mathcal{N}(0,1)}(x^{2j}) = (2j-1)!!$ it is easy to evaluate 
\begin{equation}
    \mathbb{E}_{x \sim \mathcal{N}(0,1)} (p(x^2)) = \sum_{j=0}^{15}c_j = 0.1222.
\end{equation}
This proves the claim. Note that the bound could be mildly improved by taking smaller $\delta$ and larger $L$.
\end{proof}

\begin{corollary}
The quantity %
\begin{equation}\label{eq:naiveprob1}
p_0 = \liminf_{n\rightarrow \infty}\left(\max_{j=0,1}\left( \Pr_{f \leftarrow \mathcal{F}_n^{\text{prom}}}(S(f)=j)\right)\right)
\end{equation}
satisfies $p_0 \leq 11/12$ for the \texttt{poly3-SGAP} problem.
\end{corollary}
\begin{proof}
The previous theorem states that at least 0.12 of the probability mass lies in the set of NO instances $(S(f)=0)$ for sufficiently large $n$. Meanwhile, as discussed in Lemma \ref{lem:promisefraction}, at least 1/12 of the probability mass lies in the set of YES instances $(S(f)=1)$ for all $n$. In other words, the fraction of NO instances is at most 11/12, and the fraction of YES instances is at most 0.88, which proves the corollary. 
\end{proof}

\bibliographystyle{abbrvnat}
\bibliography{references}

\end{document}